\documentclass[3p, review]{elsarticle}
\usepackage{amsmath}
\usepackage{amsthm}
\usepackage{amssymb}

\usepackage{subfig}
\usepackage{graphicx}
\usepackage{booktabs}
\usepackage{braket}
\usepackage{enumerate}
\usepackage[title]{appendix}
\usepackage[colorlinks=true, allcolors=blue]{hyperref}
\usepackage{pgfplots}
\newtheorem{proposition}{Proposition}
\newtheorem{lemma}{Lemma}
\newtheorem{theorem}{Theorem}
\newtheorem{remark}{Remark}

\usepackage{algorithmic}
\usepackage{algorithm}
\usepackage[most]{tcolorbox}
\newtcbtheorem{mytcbprob}{Problem}{}{problem}
\usepackage{diagbox}
\usepackage{makecell}
\usepackage{amssymb}
\usepackage{bm}
\usepackage{multirow}
\usepackage{float}
\allowdisplaybreaks

\title{Distributed  exact multi-objective quantum search algorithm}
\begin{document}
\begin{frontmatter}
\author{Hao Li $^{\rm a,b,c,d}$}
\author{Daowen Qiu $^{\rm a,b,c,d*}$}
\author{Le Luo $^{\rm c,d,e}$}
\address{$^{\rm a}$ Institute of Quantum Computing and Software, School of Computer Science and Engineering, \\ Sun Yat-sen University, Guangzhou 510006, China;}
\address{$^{\rm b}$The Guangdong Key Laboratory of Information Security Technology, \\Sun Yat-sen University, Guangzhou, 510006, China;}
\address{$^{\rm c}$Shenzhen Research Institute of Sun Yat-Sen University, Nanshan Shenzhen, 518057, China;}
\address{$^{\rm d}$ QUDOOR Technologies Inc., China;}
\address{$^{\rm e}$ School of Physics and Astronomy, Sun Yat-sen University, Zhuhai 519082, China}
\cortext[mycorrespondingauthor]{Corresponding author.  E-mail address:  issqdw@mail.sysu.edu.cn (D. Qiu).}

\begin{abstract}

 Multi-objective search  means searching for any one of several objectives in an unstructured database.  Grover's algorithm has quadratic acceleration in multi-objection search than classical ones.  Iterated operator in  Grover's algorithm is a key element and plays an important role in amplitude amplification. 
 In this paper, we design two distributed iterated operators and therefore two new distributed Grover's algorithms are obtained with the following advantages:  (1) Compared to   Grover's algorithm and  the modified Grover’s algorithm by Long, our  distributed algorithms  require fewer qubits; (2) Compared to the distributed Grover's algorithm proposed by Qiu et al., one of our distributed  algorithms  is exact.  Of course,   both our distributed algorithms require quite quantum communication and involve a number of more complicated unitary operators as cost,  but  there still may  have  certain advantage of  physical realizability in the  Noisy Intermediate-Scale Quantum (NISQ) era.

\end{abstract}

\end{frontmatter}

\section{Introduction}

\subsection{Background and related work}

Quantum algorithms have shown great power in important research problems such as factorizing large number  \cite{shor_polynomial-time_1997}, unstructured database search \cite{grover_fast_1996} and solving systems of linear equations \cite{HHL_2009}. However, the development of large-scale universal  quantum computers is still slow due to the current physics technologies. In the NISQ era \cite{preskill_quantum_2018}, one can implement quantum algorithms with fewer qubits and shallower circuit depth. Therefore,  it is  fascinating to design new quantum algorithms  that require fewer qubits and even shallower circuit depth compared with previous quantum algorithms for solving the same problems.

Distributed quantum computing is a  novel computing architecture that skillfully combines quantum computing and distributed  computing, and   designing distributed quantum algorithms is an intriguing and core problem in distributed quantum computing \cite{beals_efficient_2013,Qiu2017DQC,
LeGall2018,LeGall2019,avron_quantum_2021,Tan2022DQCSimon,
Qiu24,Xiao2023DQAShor,Xiao2023DQAkShor,Hao2023DDJ,
Hao2023DGSP,Zhou2023DEGrover,Zhou2023DBV}.  
In general, centralized quantum algorithms require a large number of qubits to be implemented in practical application, but distributed quantum algorithms  consists of  multiple computing nodes together with quantum  communication, which likely need smaller number of qubits to collaboratively complete the computational problem. Meanwhile, distributed quantum algorithms may also reduce the circuit depth compared to a single computing node, and is thus conducive to decreasing the impact of noise. Therefore, in the NISQ era, distributed quantum algorithms are more likely to be physically realized and practically applied.

We briefly review related contributions concerning distributed quantum algorithms.
In 2013, Beals et al. \cite{beals_efficient_2013} proposed an interesting parallel addressing algorithm for quantum memories and our algorithms in this paper also employ this algorithm for realizing some complicated query operators.  In 2017, Li and Qiu et al. \cite{Qiu2017DQC} proposed a distributed quantum phase estimation algorithm. In 2018, Le Gall et al. \cite{LeGall2018} investigated the  distributed quantum algorithms  in the quantum CONGEST model.  In 2019, Le Gall et al. \cite{LeGall2019} showed the superiority of  distributed quantum computing in the LOCAL model. In 2021, Avron et al. \cite{avron_quantum_2021} proposed a method 
  to decompose  an oracle into multiple sub-oracles. 
 In 2022, Qiu et al. \cite{Qiu24}  proposed a distributed Grover's algorithm that can reduce the number of qubits compared to Grover's algorithm; Tan, Xiao and Qiu et al. \cite{Tan2022DQCSimon} proposed an interesting and novel distributed quantum algorithm for  Simon's problem.

  These distributed quantum algorithms by Qiu and Tan et al \cite{Qiu24, Tan2022DQCSimon}  also presented a kind of  technical method for designing other distributed quantum algorithms, and  it further promoted and inspired that Li and Qiu et al. \cite{Hao2023DDJ,Hao2023DGSP} designed distributed exact quantum algorithms for solving Deutsch-Jozsa problem and  generalized Simon’s problem.  By  using  a different  technical method of designing distributed quantum algorithms, in 2023, Xiao and Qiu et al. proposed a new distributed Shor's algorithm \cite{Xiao2023DQAShor,Xiao2023DQAkShor}. 
 Then, motivated by  \cite{Qiu24}, Zhou and Qiu  et al. \cite{Zhou2023DEGrover,Zhou2023DBV} further proposed a distributed exact Grover’s algorithm   and a distributed Bernstein–Vazirani algorithm. By the way,  the distributed exact Grover's algorithm proposed by Zhou and Qiu et al. \cite{Zhou2023DEGrover} can only solve the search problem in the single-objective case, and  the complexity of computing some sub-functions may be high in this algorithm \cite{Zhou2023DEGrover}.

Grover’s algorithm was proposed   in 1996 \cite{grover_fast_1996}, which is one of the most important  algorithms in quantum computing, as it provides a quadratic speedup for the unstructured search problem.  Given a Boolean function $f:\{0,1\}^n \rightarrow \{0,1\}$, where $|\{x\in\{0,1\}^n| f(x)=1\}|=a\geq 1$, $N=2^n$, to search for an $x\in\{0,1\}^n$ with $f(x)=1$, by Grover’s algorithm we can get the objective with $\left\lfloor \pi/\left(4 \arcsin {\sqrt{a/N}}\right)\right\rfloor$ queries, and the success probability is $\sin^2{\left(\left(2\left\lfloor \pi/\left(4 \arcsin {\sqrt{a/N}}\right)\right\rfloor+1\right)\arcsin {\sqrt{a/N}}\right)}$. 
Moreover, the query complexity of Grover’s algorithm is $\Theta(\sqrt{N})$, which  has been proved to be optimal \cite{Bennett1997,Boyer1998,Zalka1999}. In fact, any classical algorithm to solve it requires $\Omega(N)$ queries. In addition, the problem of operator coherence  dynamics in Grover’s algorithm was considered (for example, see \cite{Pan2019} and the references therein).

In 2001, Grover’s algorithm with zero theoretical failure rate was considered by Long \cite{Long2001}. The modification is to replace the phase inversion in Grover's algorithm by phase rotation. The rotation angle is given analytically to be $\varphi=2\arcsin\left(\sin\left(\frac{\pi}{4J+6}\right) / \sin \theta\right)$, where $J=\left\lfloor(\pi/2-\theta)/(2\theta)\right\rfloor$, and $\theta= \arcsin {\sqrt{a/N}}$. After the $(J+1)$-th iteration, the objective state is obtained with certainty by measurement. Long gave a detailed derivation of this result in the $SO(3)$ picture  of the quantum search algorithm, which exploited  the relationship between $SU(2)$ and $SO(3)$   \cite{Long2001SO3}. In addition, there are two other methods  to make Grover's algorithm exact \cite{Brassard2000, Hoyer2000}.

The distributed Grover's algorithm  proposed by Qiu et al. \cite{Qiu24} has a linear advantage in quantum query complexity compared to Grover's algorithm, and particularly, the number of qubits in the algorithm designed by Qiu et al. \cite{Qiu24} has been reduced to $n-t$, versus $n$ in Grover's algorithm.  
Also, Qiu et al. \cite{Qiu24} proposed an algorithm for constructing quantum circuits to implement the oracle corresponding to any Boolean function with conjunctive normal form (CNF), providing a physically realizable method for some quantum query algorithms in practice.


\subsection{Contributions and paper structure}

The distributed Grover's algorithm designed by Qiu et al. \cite{Qiu24} can solve the search problem in multi-objecitve case, and its success probability is high but not one, so it still is not completely exact.    
The search problem in the multi-objective case refers to searching for one of multiple objectives in an unstructured database.
 One question now naturally arises as follows.

\textbf{The question:}  Is it possible to design a distributed quantum algorithm that can exactly solve the search problem in the case of multi-objective?

In this paper, we answer the above question affirmatively.
Specifically, we design two distributed quantum algorithms  by means of designing two  distributed iterated operators, one of which
 can exactly solve the search problem in the case of  multi-objective. Both our distributed quantum algorithms have the following advantages and disadvantages:
\begin{itemize}
\item 
Compared to  Grover's algorithm \cite{grover_fast_1996} and  the  modified Grover’s algorithm by Long \cite{Long2001},  both our distributed quantum algorithms require fewer qubits for the largest single computing node and one of our algorithms is exact.
\item 
 Compared to the distributed Grover's algorithm proposed by Qiu et al. \cite{Qiu24}, one of our distributed quantum algorithms  is exact.

\item Our distributed quantum algorithms require quite quantum communication and involve a number of some complicated unitary operators as cost, but none of  Grover's algorithm \cite{grover_fast_1996} and  the  modified Grover’s algorithm by Long \cite{Long2001} as well as  the distributed Grover's algorithm proposed by Qiu et al. \cite{Qiu24} need quantum communication.
\end{itemize}

The rest of this paper is organized as follows.  First,
in Section \ref{Sec2}, we recall  Grover's algorithm and the modified Grover’s algorithm by Long. Then,
in Section \ref{Sec3}, we describe  and analyze a distributed  multi-objective quantum search algorithm by designing a distributed iterated operator. After that,
in Section \ref{Sec4}, we describe  and analyze a distributed exact multi-objective quantum search algorithm by designing another distributed iterated operator.  In addition, comparisons to other relevant algorithms are made in Section \ref{Sec5},  and finally, in Section \ref{Sec6}, we conclude with a summary and mention some problems for further consideration.

\section{Preliminaries}\label{Sec2}

In the interest of readability, we briefly review  Grover's algorithm \cite{grover_fast_1996} and the modified Grover’s algorithm by Long \cite{Long2001}.

\subsection{Grover's algorithm \cite{grover_fast_1996}}

 Grover's algorithm aims to search for a goal (solution) in a very
wide range. We formally describe this problem  as follows.  

\begin{mytcbprob*}{the search problem}
\textbf{Input:} A Boolean function $f:\{0,1\}^n \rightarrow\{0,1\}$, $N=2^n$.

\textbf{Promise:} $|\{x\in\{0,1\}^n| f(x)=1\}|=a\geq 1$.

\textbf{Output:} An input $x\in\{0,1\}^n$ such that $f(x) = 1$.
\end{mytcbprob*}

Grover's algorithm performs the search quadratically faster than any classical algorithm.  In fact, any classical algorithm
finding a solution with probability at least $\frac{2}{3}$ must make
$\Omega(N)$ queries in the worst case, but Grover's  algorithm takes only
$O(\sqrt{N})$ queries, more exactly $\left\lfloor \pi/\left(4 \arcsin {\sqrt{a/N}}\right)\right\rfloor$ queries.

\begin{algorithm}[H]
	\caption{Grover's algorithm}
	\begin{algorithmic}\label{Grover.alg}
	\STATE \textbf{Input}: A function $f:\{0,1\}^n \rightarrow\{0,1\}$ with $|\{x\in\{0,1\}^n| f(x)=1\}|=a\geq 1$, $N=2^n$.
	\STATE \textbf{Output}: The string 
	$x$ such that $f(x)=1$ with probability $\sin^2{\left(\left(2\left\lfloor \pi/\left(4 \arcsin {\sqrt{a/N}}\right)\right\rfloor+1\right)\right.}$ ${\left.\arcsin {\sqrt{a/N}}\right)}$.
	
\STATE \textbf{Procedure:}

	\STATE \textbf{1.}	 $H^{\otimes n}|0\rangle^{\otimes n}=\frac{1}{\sqrt{2^n}}\sum\limits_{x\in\{0 , 1\}^n}|x\rangle$.
	
	 \STATE \textbf{2.}	$G$ is performed with  $\left\lfloor \pi/\left(4 \arcsin {\sqrt{a/N}}\right)\right\rfloor$ times, where $G=-H^{\otimes n}Z_0H^{\otimes n}Z_f$, 
	 
	 \qquad\qquad$Z_f|x\rangle=(-1)^{f(x)}|x\rangle$, $Z_0|x\rangle=\left\{
\begin{array}{rcl}
-|x\rangle, & &x=0^n ;\\
|x\rangle, & & x\neq 0^n.
\end{array} \right.$
	
	\STATE \textbf{3.}	Measure the resulting state.
	\end{algorithmic}
\end{algorithm}



\begin{proposition}[See \cite{Qiu24}]\label{pro_Grover}
Let Boolean function $f:\{0,1\}^n \rightarrow\{0,1\}$ with $|\{x\in\{0,1\}^n| f(x)=1\}|=a\geq 1$, $N=2^n$. Then the  string $x$ such that $f(x)=1$ can be  obtained by  Algorithm  \ref{Grover.alg} querying $\left\lfloor \pi/\left(4 \arcsin {\sqrt{a/N}}\right)\right\rfloor$ times  with probability $\sin^2{\left(\left(2\left\lfloor \pi/\left(4 \arcsin {\sqrt{a/N}}\right)\right\rfloor+1\right)\arcsin {\sqrt{a/N}}\right)}$.
\end{proposition} 

\subsection{The modified Grover’s algorithm by Long \cite{Long2001}}
In 2001, Long improved Grover’s algorithm and proposed a modified version of Grover’s algorithm, which will acquire the goal state with certainty. Its main idea is to extend the iterated operator  $G$ to operator $L$. 
The adjustment is carried out by two steps: (1) substitute phase rotation for phase inversion; (2) iterate $J+1$ times for operator $L$, where $J=\left\lfloor(\pi/2-\theta)/(2\theta)\right\rfloor$ and $\theta= \arcsin {\sqrt{a/N}}$. 

\begin{algorithm}[H]
\caption{The modified Grover’s algorithm by Long}
\label{Long.alg}
\begin{algorithmic}
\STATE \textbf{Input}: A function $f:\{0,1\}^n \rightarrow\{0,1\}$ with $|\{x\in\{0,1\}^n| f(x)=1\}|=a\geq 1$, $N=2^n$.
\STATE \textbf{Initialization}: $\theta= \arcsin {\sqrt{a/N}}$, $J=\lfloor(\pi/2-\theta)/(2\theta)\rfloor$, $\varphi=2\arcsin\left(\sin\left(\frac{\pi}{4J+6}\right) / \sin \theta\right)$.
\STATE \textbf{Output}: The string $x$ such that $f(x)=1$ with certainty.
\STATE \textbf{Procedure:}
\STATE \textbf{1.} $H^{\otimes n}|0\rangle^{\otimes n}=\frac{1}{\sqrt{2^n}}\sum\limits_{x\in\{0 , 1\}^n}|x\rangle$.
\STATE \textbf{2.}  $L$ is   performed with  $J+1$ times, where $L= -H^{\otimes n} R_0(\varphi) H^{\otimes n}R_f(\varphi)$, 

 \qquad\qquad$R_f(\varphi)|x\rangle=e^{\mathrm{i}\varphi\cdot f(x)}|x\rangle$, $R_0(\varphi)|x\rangle=\left\{
\begin{array}{rcl}
e^{\mathrm{i}\varphi}|x\rangle, & &x=0^n ;\\
|x\rangle, & & x\neq 0^n.
\end{array} \right.$

\STATE \textbf{3.} Measure the resulting state.
\end{algorithmic}
\end{algorithm}

\begin{proposition}[See \cite{Long2001}]\label{pro_Long}
Let Boolean function $f:\{0,1\}^n \rightarrow\{0,1\}$ with $|\{x\in\{0,1\}^n| f(x)=1\}|=a\geq 1$, $N=2^n$. Then the string $x$ with $f(x)=1$ can be exactly obtained by  Algorithm \ref{Long.alg} querying $J+1$ times, where  $J=\lfloor(\pi/2-\theta)/(2\theta)\rfloor$, $\theta= \arcsin {\sqrt{a/N}}$.
\end{proposition}

\begin{remark}
The number of queries in  Algorithm \ref{Grover.alg} is $\left\lfloor \pi/\left(4 \arcsin {\sqrt{a/N}}\right)\right\rfloor$,  the number of queries in Algorithm \ref{Long.alg} is $
\left\lfloor \pi/\left(4 \arcsin {\sqrt{a/N}}\right)-1/2\right\rfloor+1$. Let $\pi/\left(4 \arcsin {\sqrt{a/N}}\right)=\left[ \pi/\left(4 \arcsin {\sqrt{a/N}}\right)\right]+\left\{\pi/\left(4 \arcsin {\sqrt{a/N}}\right)\right\}$, where $\left\{\pi/\left(4 \arcsin {\sqrt{a/N}}\right)\right\}$ denotes the fractional part of $\pi/\left(4 \arcsin {\sqrt{a/N}}\right)$.  If $0\leq \left\{\pi/\left(4 \arcsin {\sqrt{a/N}}\right)\right\}<1/2$, the number of queries in Algorithm \ref{Long.alg} is  is equal to the number of queries in Algorithm \ref{Grover.alg}. If $1/2\leq \left\{\pi/\left(4 \arcsin {\sqrt{a/N}}\right)\right\}<1$, the number of queries in Algorithm \ref{Long.alg} is  one more time than the number of queries in Algorithm \ref{Grover.alg}.
\end{remark}

\section{Distributed  multi-objective quantum search algorithm}\label{Sec3}

In this section, we first describe  the design of  our distributed  multi-objective quantum search algorithm, and then give its correctness, space complexity and communication complexity analyses.

\subsection{Design of algorithm}

Let Boolean function $f:\{0,1\}^n\rightarrow\{0,1\}$, and suppose  $|\{x\in\{0,1\}^n| f(x)=1\}|=a\geq 1$. First, as the distributed Grover's algorithm designed by Qiu et al. \cite{Qiu24}, 
 $f$ is divided into $2^t$  subfunctions $f_w:\{0,1\}^{n-t}\rightarrow\{0,1\}$, where $f_w(u)=f(uw)$, $u \in \{0,1\}^{n-t}$, $w \in \{0,1\}^t$, $1\leq t<n$. 
The search problem in distributed scenario is to find  an input $x\in\{0,1\}^n$ such that $f(x) = 1$  by querying the $2^t$ subfunctions $f_w$.

In the following, we  describe the  operators and notations used in our algorithm, i.e., Algorithm \ref{DMA.alg}.

Define  the operator $Z_{f}': \{0,1\}^{n+2^t+1}\rightarrow  \{0,1\}^{n+2^t+1}$  as:
\begin{equation}\label{Z_{f}'}
Z_{f}'\Ket{u,w,a,b}=(-1)^{b\oplus f_w(u)}\Ket{u,w,a,b},
\end{equation}
where   $u\in\{0,1\}^{n-t}$, $w\in\{0,1\}^{t}$, $a=a_{0^t}a_{0^{t-1}1}\cdots a_{1^t}\in \{0,1\}^{2^t}$ and $b\in \{0,1\}$.

Intuitively, the effect of  $Z_f'$ on the quantum state is similar to  $Z_f$ defined above. 

Denote
\begin{align}
&A=\{x\in\{0 , 1\}^n | f(x)=1\} 
,\\
&B=\{x\in\{0 , 1\}^n | f(x)=0\}
,\\
&|A'\rangle=\frac{1}{\sqrt{|A|}}\sum\limits_{x\in A}|x\rangle\Ket{0}^{\otimes (2^t+1)},\\
&|B'\rangle=\frac{1}{\sqrt{|B|}}\sum\limits_{x\in B}|x\rangle\Ket{0}^{\otimes (2^t+1)}.
\end{align}


Denote $a=|A|$, $b=|B|$, $N=2^n$, $\sin\theta=\sqrt{\frac{a}{N}}$,  $\cos\theta=\sqrt{\frac{b}{N}}$, and  
\begin{equation}\label{h}
\begin{split}
\Ket{h}=\sin\theta\Ket{A'}+\cos\theta\Ket{B'}.
\end{split}
\end{equation}

Define the function $ \rm OR$ as:
\begin{equation}
\begin{split}
		 {\rm OR}(x)=
		\begin{cases}
		1, & |x|\geq1\text{;}\\
		0,& |x|=0\text{,}
          \end{cases}
\end{split}
\end{equation}
where $|x|$ is the Hamming weight of $x$.

Define the function $g: \{0,1\}^{n+2}\rightarrow \{0,1\}$ as: 
\begin{equation}\label{g(u,w,b,c)}
\begin{split}
g(u,w,b,c)=c\oplus\left(\lnot {\rm OR}(w(b\oplus {\rm OR}(u)))\right),
\end{split}
\end{equation}
where  $u\in\{0,1\}^{n-t}$, $w\in\{0,1\}^{t}$, $b\in \{0,1\}$ and  $c\in \{0,1\}$.

Denote   
\begin{equation}\label{h^perp}
\begin{split}
\Ket{h^{\perp}}=\sum_{\substack{bcd\neq 0^{2^t+1}\\ g(u,w,b,c)=1 }}H^{\otimes n}\Ket{u,w}\Ket{b,c,d}\Bra{u,w}H^{\otimes n}\Bra{b,c,d},
\end{split}
\end{equation}
where   $u\in\{0,1\}^{n-t}$, $w\in\{0,1\}^{t}$, $b\in \{0,1\}$, $c\in \{0,1\}$  and  $d\in \{0,1\}^{2^t-1}$.


Define the operator $Z_{0,H}': \{0,1\}^{n+2^t+1}\rightarrow  \{0,1\}^{n+2^t+1}$  as:
\begin{equation}\label{Z_{0,H}'}
\begin{split}
Z_{0,H}'=
I^{\otimes(n+2^t+1)}-2\Ket{h}\Bra{h}-2\Ket{h^{\perp}}\Bra{h^{\perp}}.
\end{split}
\end{equation}

Define the operator $G': \{0,1\}^{n+2^t+1}\rightarrow  \{0,1\}^{n+2^t+1}$ as:
\begin{equation}\label{G'}
\begin{split}
G'=-Z_{0,H}'Z_{f}'.
\end{split}
\end{equation}

\begin{figure*}[h]
  \centerline{\includegraphics[width=\textwidth]{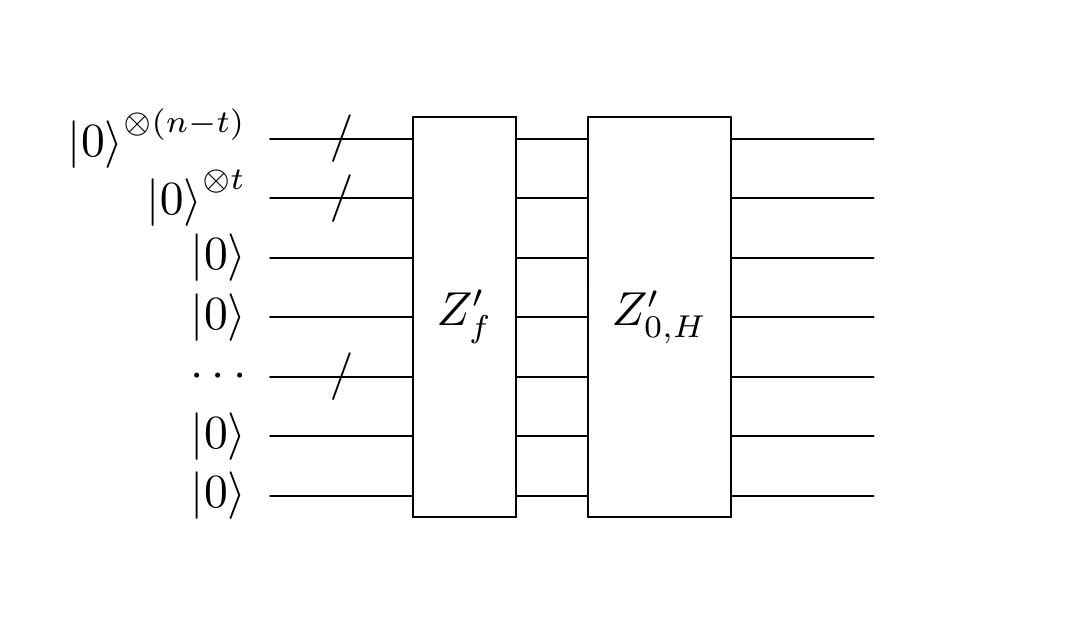}}
  \caption{The circuit for  operator $-G'$.}
\end{figure*}

Indeed, both  $Z_{f}'$ and $Z_{0,H}'$ can be decomposed into a number of smaller and physically realizable operators. 

In the following, we introduce some  notations and operators related to the decomposition representation of  $Z_{f}'$.
Let $[N(t)]$ represent the set of integers $\{0,1,\cdots, 2^t-1\}$, and let ${\rm BI}:\{0,1\}^t \rightarrow [N(t)]$ be the function to convert a binary string of $t$ bits to an equal decimal integer. For any operator $A_w$ with $w\in \{0,1\}^t$, let
$\prod\nolimits_{w\in\{0,1\}^t}A_w= A_{1^t}A_{1^{t-1}0}{A_{1^{t-2}01}}\cdots{A_{0^{t-1}1}} A_{0^t}$.

Define the  query operators $O_{f_w}: \{0,1\}^{n+{\rm BI}(w)+1}\rightarrow \{0,1\}^{n+{\rm BI}(w)+1}$
as:
\begin{equation}\label{O_{f_w}}
O_{f_w}\ket{u}\ket{b_w}\ket{c}=\ket{u}\ket{b_w}\ket{c\oplus f_w(u)},
\end{equation}
where $u\in\{0,1\}^{n-t}$, $w\in \{0,1\}^t$, $b_w\in\{0,1\}^{t+{\rm BI}(w)}$ and $c\in \{0,1\}$.

\begin{figure*}[h]
  \centerline{\includegraphics[width=0.5\textwidth]{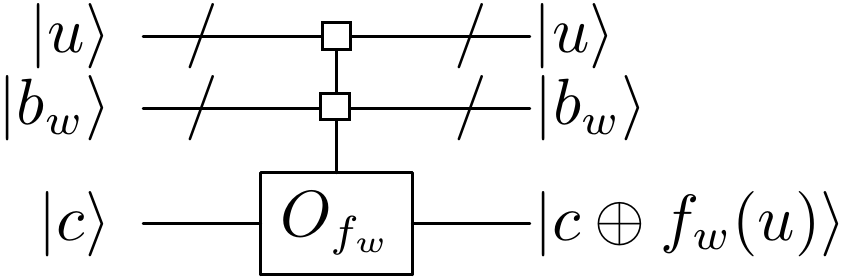}}
  \caption{The circuit for  operator $O_{f_w}$.}
\end{figure*}

Define the   operator $O^*_{f}: \{0,1\}^{n+2^t+1}\rightarrow \{0,1\}^{n+2^t+1}$
as:
\begin{equation}\label{O^*_{f}}
O^*_{f}=\prod\nolimits_{w\in\{0,1\}^t}\left(O_{f_{w}}\otimes I^{\otimes \left(2^t-{\rm BI}(w)\right)}\right).
\end{equation}

From Eq. (\ref{O_{f_w}}) and Eq. (\ref{O^*_{f}}), we have:
\begin{equation}\label{O_{f_w}A_w}
\begin{split}
O^*_{f}\Ket{u,w}\Ket{a}\Ket{b}=&\left(\prod\nolimits_{w\in\{0,1\}^t}\left(O_{f_{w}}\otimes I^{\otimes \left(2^t-{\rm BI}(w)\right)}\right)\right)\Ket{u,w}\Ket{a}\Ket{b}\\
=&\left(O_{f_{1^t}}\otimes I\right)\cdots\left(O_{f_{0^{t-1}1}}\otimes I^{\otimes \left(2^t-1\right)}\right)\left(O_{f_{0^{t}}}\otimes I^{\otimes 2^t}\right)\Ket{u,w}\Ket{a_{0^t}}\Ket{{a_{0^{t-1}1}}}\cdots\Ket{a_{1^t}}\Ket{b}\\
=&\left(O_{f_{1^t}}\otimes I\right)\cdots\left(O_{f_{0^{t-1}1}}\otimes I^{\otimes \left(2^t-1\right)}\right)\Ket{u,w}\ket{a_{0^t}\oplus f_{0^t}(u)}\Ket{a_{0^{t-1}1}}\cdots\Ket{a_{1^t}}\Ket{b}\\
=&\left(O_{f_{1^t}}\otimes I\right)\cdots\left(O_{f_{0^{t-1}1}}\otimes I^{\otimes \left(2^t-1\right)}\right)\Ket{u,w}\ket{a_{0^t}\oplus f_{0^t}(u)}\Ket{a_{0^{t-1}1}}\cdots\Ket{a_{1^t}}\Ket{b}\\
=&\Ket{u,w}\ket{a_{0^t}\oplus f_{0^t}(u)}\Ket{a_{0^{t-1}1}\oplus f_{0^{t-1}1}(u)}\cdots\Ket{a_{1^t}\oplus f_{1^t}(u)}\Ket{b},
\end{split}
\end{equation}
where   $u\in\{0,1\}^{n-t}$, $w\in\{0,1\}^{t}$, $a=a_{0^t}a_{0^{t-1}1}\cdots a_{1^t}\in \{0,1\}^{2^t}$ and $b\in \{0,1\}$.

\begin{figure*}[h]
  \centerline{\includegraphics[width=0.8\textwidth]{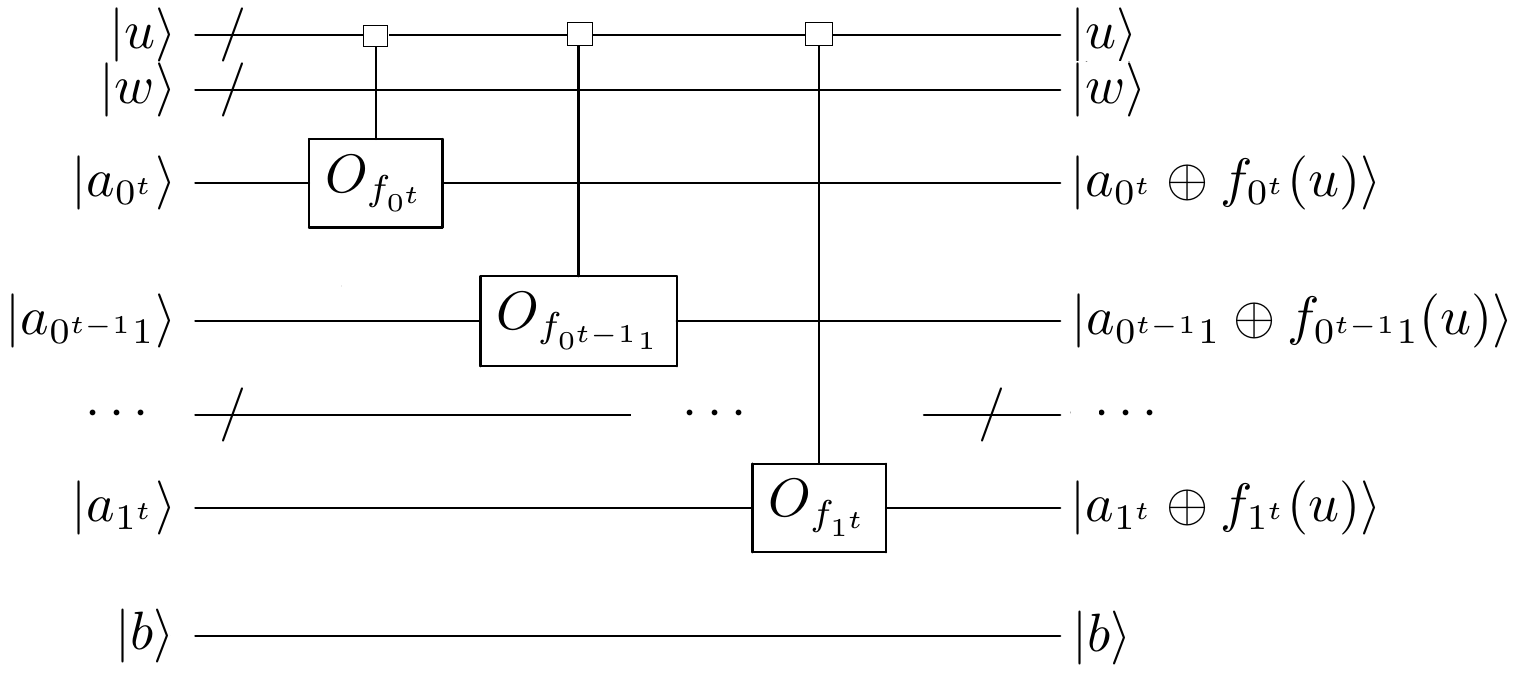}}
  \caption{The circuit for  operator $O^*_{f}$.}
\end{figure*}

Define the operator $U: \{0,1\}^{2^t+t+1}\rightarrow \{0,1\}^{2^t+t+1}$  as:
\begin{equation}\label{U}
\begin{split}
U\ket{i}\left(\bigotimes_{j\in\{0,1\}^{t}}\ket{a_j}\right)|b\rangle
=\ket{i}\left(\bigotimes_{j\in\{0,1\}^{t}}\ket{a_j}\right)\Ket{b\oplus a_i},
\end{split}
\end{equation}
where  $i\in\{0,1\}^{t}$, $a_j\in \{0,1\}$, $b\in \{0,1\}$, and
$\bigotimes\nolimits_{j\in\{0,1\}^{t}}\ket{a_j}= \ket{a_{0^t}}\ket{a_{0^{t-1}1}}\cdots \ket{a_{1^t}}$.

Intuitively, the effect of  $U$ acting on a quantum state is to query $a_i$ according to the index $i$.

\begin{figure*}[h]
  \centerline{\includegraphics[width=0.65\textwidth]{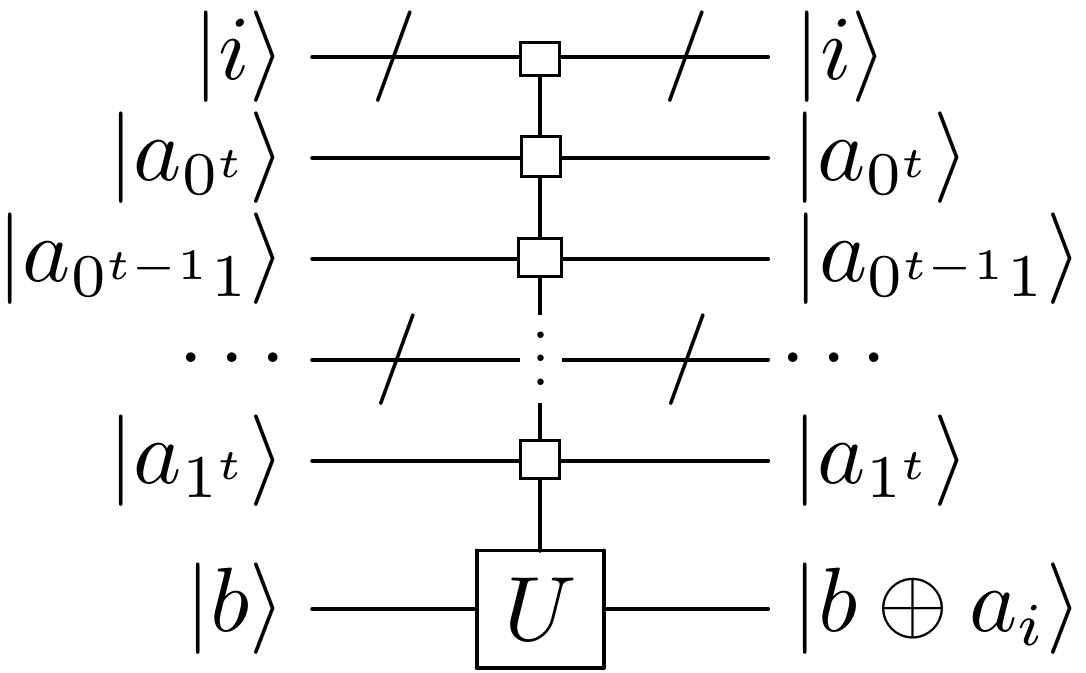}}
  \caption{The circuit for  operator $U$.}
\end{figure*}

Define the operator $F: \{0,1\}^{n+2^t+1}\rightarrow  \{0,1\}^{n+2^t+1}$  as:
\begin{equation}\label{F}
\begin{split}
F=\left(I^{\otimes(n-t)}\otimes U\right)O^*_{f}.
\end{split}
\end{equation}

Intuitively, the effect of  $F$ on a quantum state is to implement a centralized query operator. 

\begin{figure}[H]
  \centerline{\includegraphics[width=0.7\textwidth]{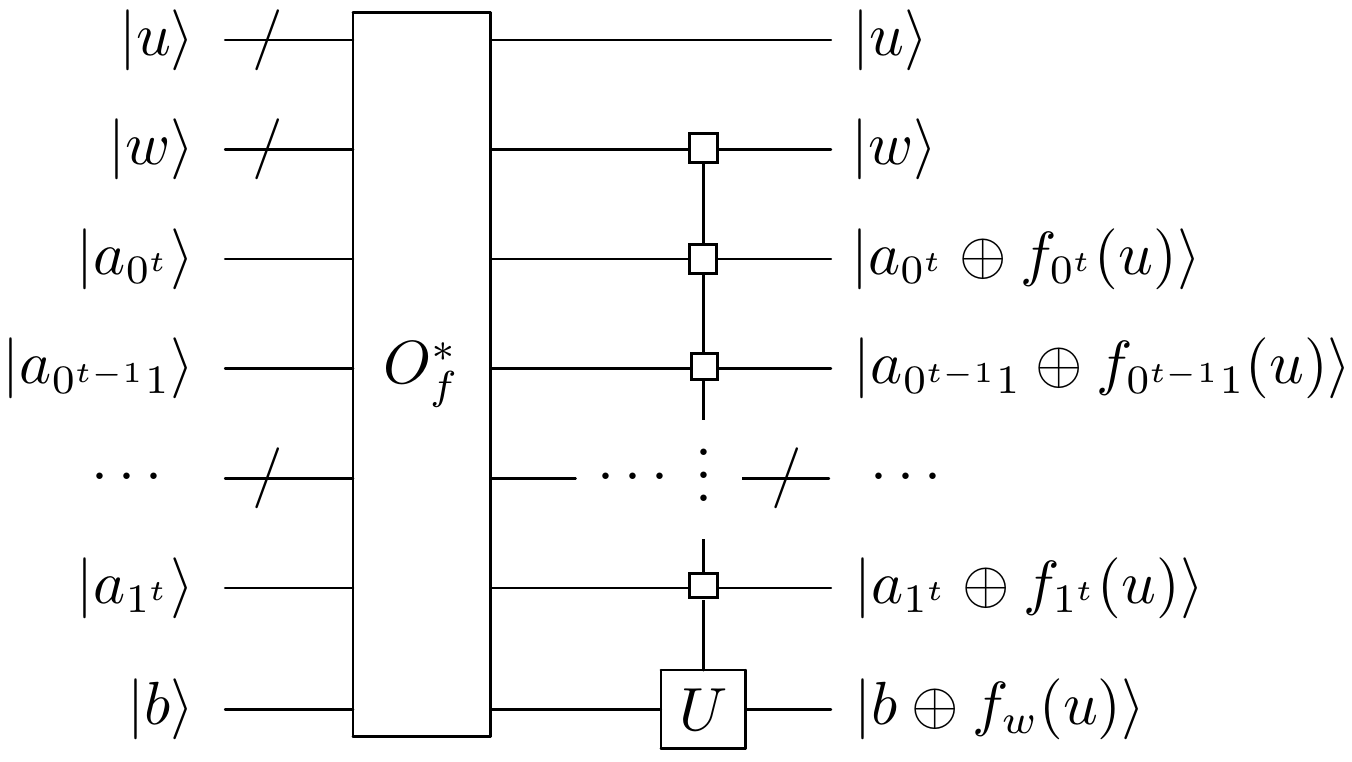}}
  \caption{The circuit for  operator $F$.}
\end{figure}

 $Z_{f}'$ defined in Eq. (\ref{Z_{f}'}) can be decomposed into a number of operators represented by the following equation:
\begin{equation}\label{Z_{f}'DE}
\begin{split}
Z_{f}'=F^{\dagger}\left(I^{\otimes \left(2^t+n\right)}\otimes Z\right)F,
\end{split}
\end{equation}
which will be proved in Proposition \ref{pp_Zf} of Subsection \ref{Correctness analysis of  algorithm}.

\begin{figure}[H]
  \centerline{\includegraphics[width=0.7\textwidth]{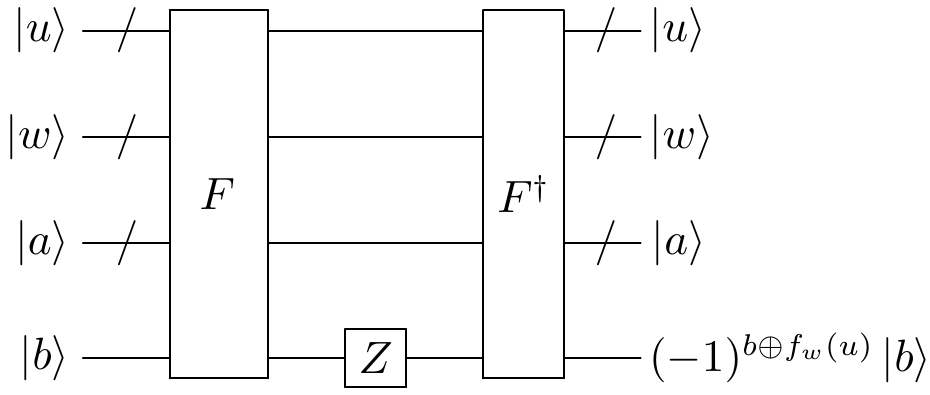}}
  \caption{The circuit for  operator $Z_{f}'$.}
\end{figure}

Next we introduce some operators and notations related to the decomposition representation of  $Z_{0,H}'$.

Define the operator $V_1: \{0,1\}^{n+1}\rightarrow \{0,1\}^{n+1}$   as: 
\begin{equation}\label{V1}
\begin{split}
V_1\ket{u}\ket{w}\ket{b}=\ket{u}\ket{w}\ket{b\oplus {\rm OR}(u)},
\end{split}
\end{equation}
where  $u\in\{0,1\}^{n-t}$, $w\in \{0,1\}^{t}$ and $b\in \{0,1\}$.

\begin{figure*}[h]
  \centerline{\includegraphics[width=0.5\textwidth]{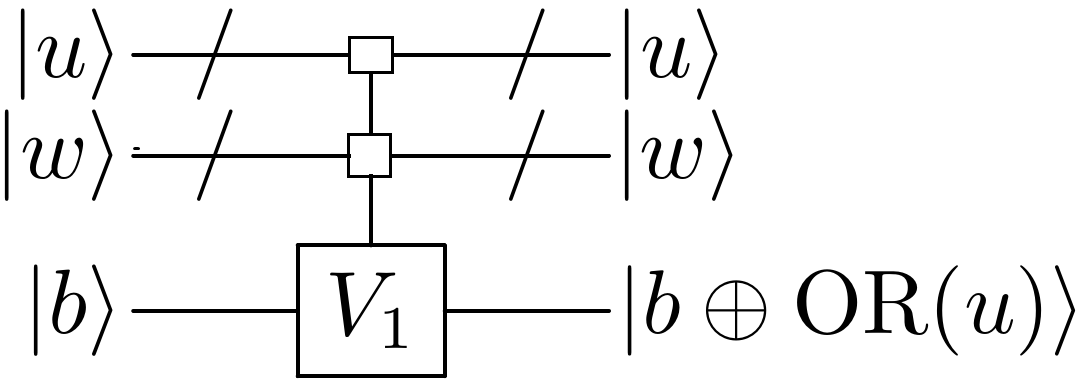}}
  \caption{The circuit for  operator $V_1$.}
\end{figure*}

Define the operator $V_2: \{0,1\}^{t+2}\rightarrow \{0,1\}^{t+2}$   as: 
\begin{equation}\label{V2}
\begin{split}
V_2\ket{w}\ket{b}\ket{c}=\ket{w}\ket{b}\ket{c\oplus \lnot {\rm OR}(wb)},
\end{split}
\end{equation}
where  $w\in\{0,1\}^{t}$, $b\in \{0,1\}$ and $c\in \{0,1\}$.

\begin{figure*}[h]
  \centerline{\includegraphics[width=0.5\textwidth]{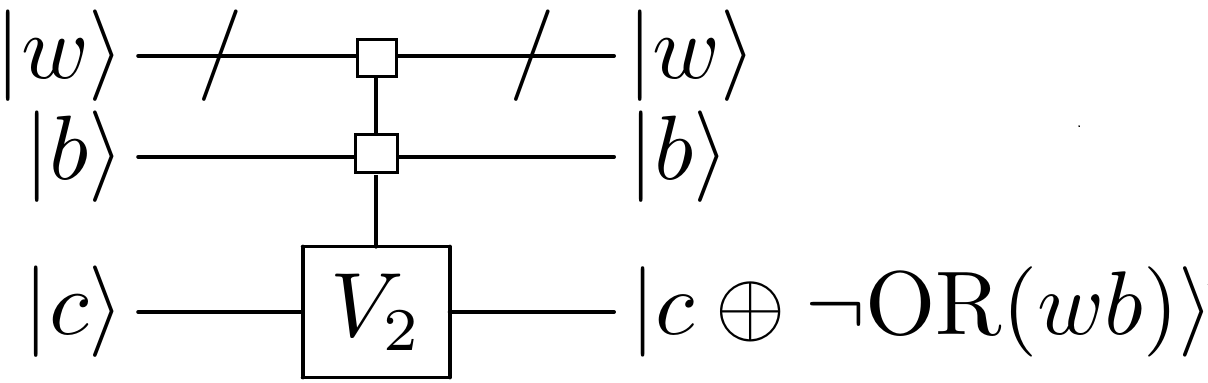}}
  \caption{The circuit for  operator $V_2$.}
\end{figure*}

Define the operator $D: \{0,1\}^{n+2^t+1}\rightarrow  \{0,1\}^{n+2^t+1}$  as:
\begin{equation}\label{D}
\begin{split}
D=\left(I^{\otimes (n-t)}\otimes V_2\otimes I^{\otimes\left(2^t-1\right)}\right)\left(V_1\otimes I^{\otimes 2^t}\right).
\end{split}
\end{equation}

\begin{figure*}[h]
  \centerline{\includegraphics[width=0.7\textwidth]{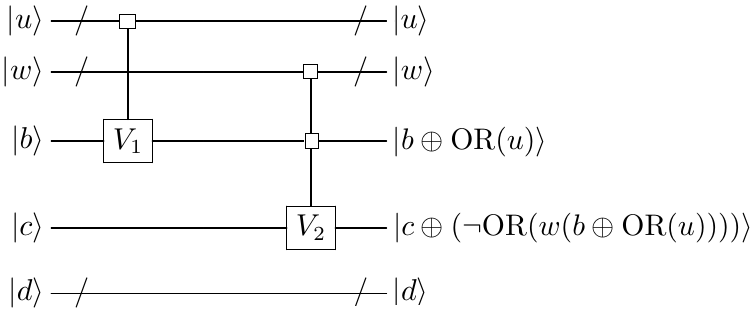}}
  \caption{The circuit for  operator $D$.}
\end{figure*}

Define the  operator $Z_0': \{0,1\}^{n+2^t+1}\rightarrow  \{0,1\}^{n+2^t+1}$  as:
\begin{equation}\label{Z_0'}
\begin{split}
Z_0'=D^{\dagger}\left(I^{\otimes (n+1)}\otimes Z\otimes I^{\otimes\left(2^t-1\right)}\right)D.
\end{split}
\end{equation}

Intuitively, the effect of  $Z_0'$ on a quantum state is similar to  $Z_0$. 

\begin{figure*}[h]
  \centerline{\includegraphics[width=0.7\textwidth]{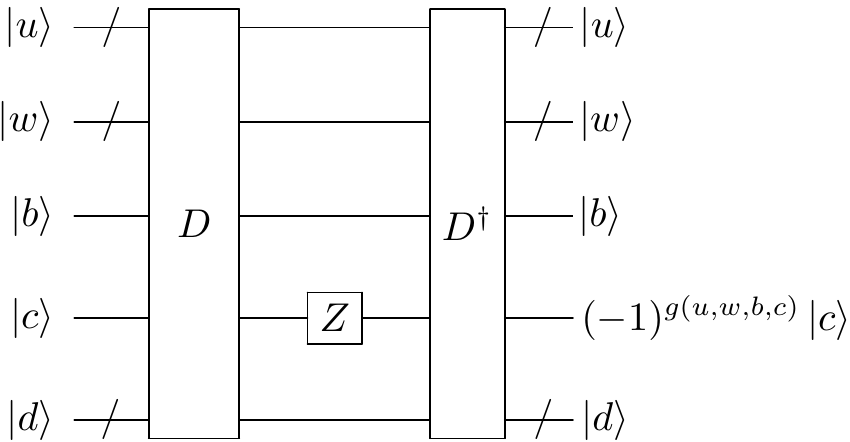}}
  \caption{The circuit for  operator $Z_0'$.}
\end{figure*}

 $Z_{0,H}'$ defined in Eq. (\ref{Z_{0,H}'}) can be decomposed into a number of operators represented by the following equation:
\begin{equation}\label{Z_{0,H}'DE}
\begin{split}
Z_{0,H}'=\left(H^{\otimes n}\otimes I^{\otimes \left(2^t+1\right)}\right)Z_0'\left(H^{\otimes n}\otimes I^{\otimes \left(2^t+1\right)}\right),
\end{split}
\end{equation}
which will be proved in Proposition \ref{pp_Z0H} of Subsection \ref{Correctness analysis of  algorithm}.

\begin{figure*}[h]
  \centerline{\includegraphics[width=0.6\textwidth]{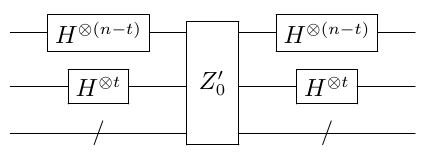}}
  \caption{The circuit for  operator $Z_{0,H}'$.}
\end{figure*}



In the following, we give the  decomposition circuit for  -$G'$.

\begin{figure}[H]
  \includegraphics[width=\textwidth]{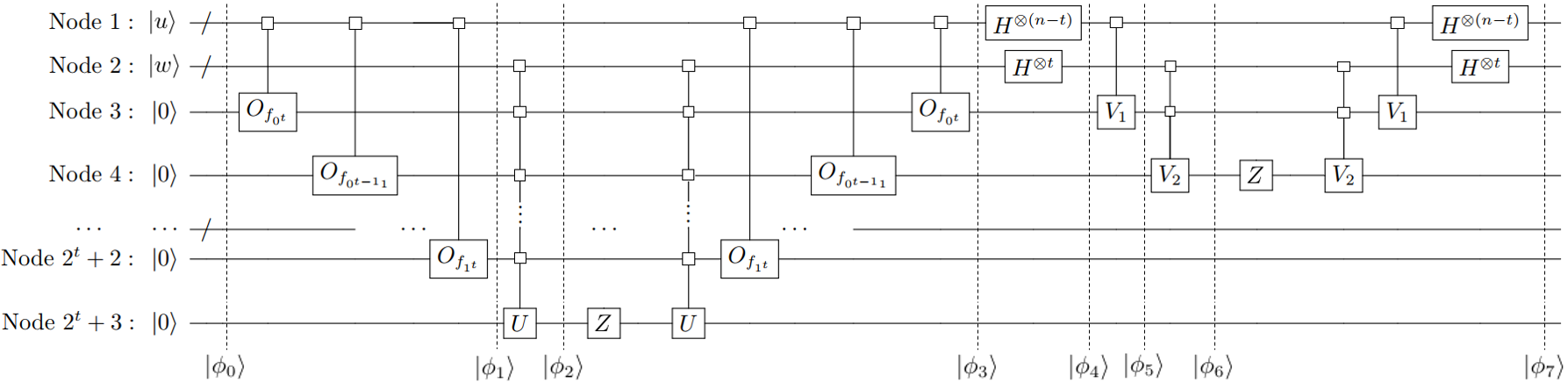}
  \caption{The decomposition circuit for    -$G'$.}
  \label{FD operator -G'}
\end{figure}

\begin{algorithm}
\caption{Distributed multi-objective quantum search algorithm}
\label{DMA.alg}
\begin{algorithmic}
\STATE \textbf{Input}:  $f:\{0,1\}^n \rightarrow\{0,1\}$ with $|\{x\in\{0,1\}^n| f(x)=1\}|=a\geq 1$, $N=2^n$.
\STATE \textbf{Initialization}: $n>4$, $1<t<\log_2n-1$.
\STATE \textbf{Output}: The string $x$ such that $f(x)=1$ with probability $\sin^2{\left(\left(2\left\lfloor \pi/\left(4 \arcsin {\sqrt{a/N}}\right)\right\rfloor+1\right)\right.}$ ${\left.\arcsin {\sqrt{a/N}}\right)}$.

\STATE \textbf{Procedure:}

\STATE \textbf{1.} $\left(H^{\otimes n}\otimes I^{\otimes \left(2^t+1\right)}\right)|0\rangle^{\otimes (n+2^t+1)}=\frac{1}{\sqrt{2^n}}\sum\limits_{x\in\{0 , 1\}^n}|x\rangle|0\rangle^{\otimes (2^t+1)}$.

\STATE \textbf{2.}  $G'$ is   performed with   $\left\lfloor \pi/\left(4 \arcsin {\sqrt{a/N}}\right)\right\rfloor$ times, where $G'=-Z_{0,H}'Z_{f}'$, $Z_f'$ is defined in Eq. ($\ref{Z_{f}'}$), and $Z_{0,H}'$ is defined in Eq. ($\ref{Z_{0,H}'}$).

\STATE \textbf{3.} Measure the first $n$ qubits of the resulting state.

\end{algorithmic}
\end{algorithm}

\begin{figure}[H]
  \includegraphics[width=\textwidth]{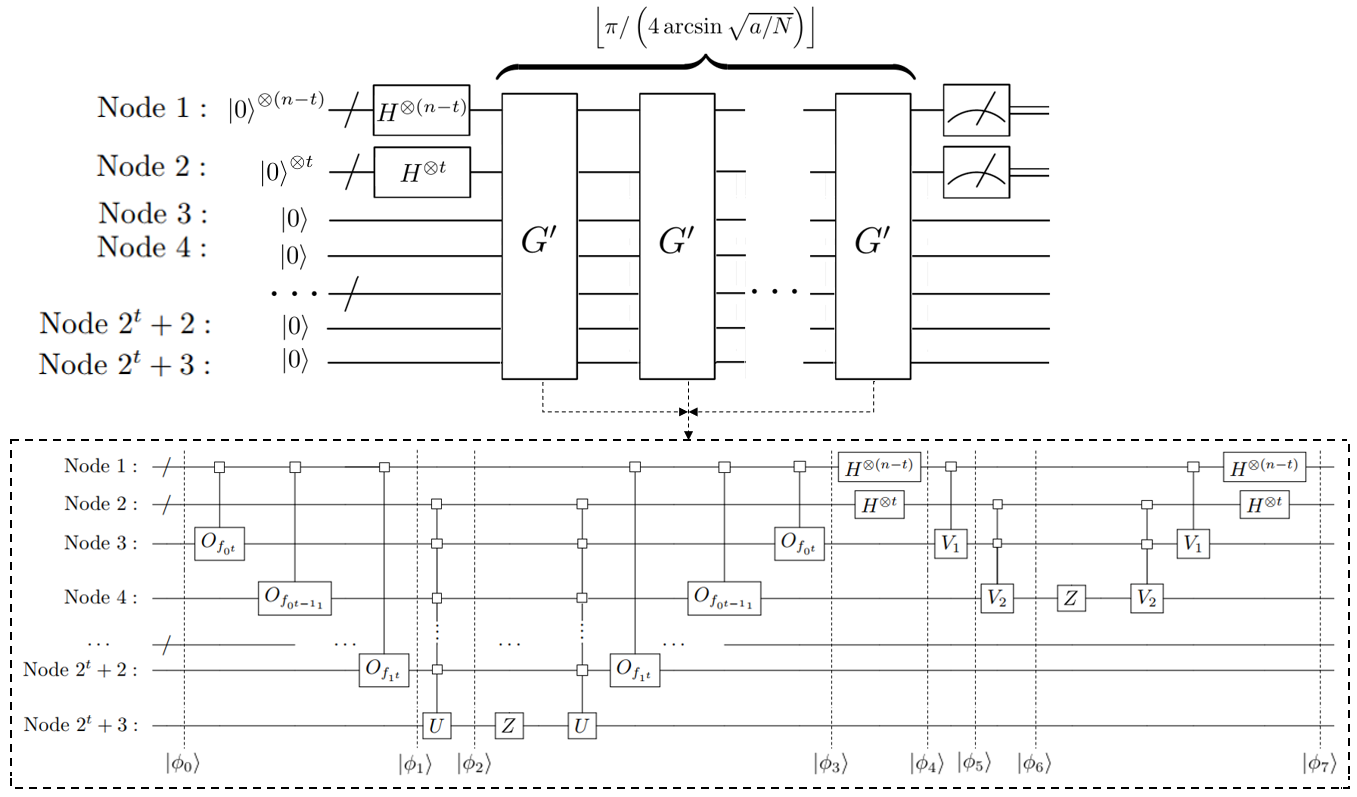}
  \caption{The circuit for  Algorithm \ref{DMA.alg}.}
\end{figure}


\subsection{Correctness analysis of  algorithm}\label{Correctness analysis of  algorithm}

In this subsection, we  first give a theorem proving the correctness of  Algorithm \ref{DMA.alg}, and then two propositions are given to show  the equivalent representations of $Z_{f}'$ and  $Z_{0,H}'$.

 
\begin{theorem}\label{Thm_DEMA1} 
Let Boolean function $f:\{0,1\}^n \rightarrow\{0,1\}$ with $|\{x\in\{0,1\}^n| f(x)=1\}|=a\geq 1$, $N=2^n$. Then the string $x$ with $f(x)=1$ can be  obtained by  Algorithm \ref{DMA.alg} querying $\left\lfloor \pi/\left(4 \arcsin {\sqrt{a/N}}\right)\right\rfloor$ times  with probability $\sin^2{\left(\left(2\left\lfloor \pi/\left(4 \arcsin {\sqrt{a/N}}\right)\right\rfloor+1\right)\arcsin {\sqrt{a/N}}\right)}$.
\end{theorem}
\begin{proof}
By Eq. (\ref{Z_{f}'}), Eq. (\ref{Z_{0,H}'}) and Eq. (\ref{G'}), we have:
\begin{equation}
\begin{split}
G'\Ket{A'}=&-Z_{0,H}'Z_{f}'\Ket{A'}\\
=&-\left(I^{\otimes(n+2^t+1)}-2\Ket{h}\Bra{h}-2\Ket{h^{\perp}}\Bra{h^{\perp}}\right)(-1)\Ket{A'}\\
=&\Ket{A'}-2\left(\sin\theta\Ket{A'}+\cos\theta\Ket{B'}\right)\left(\sin\theta\Bra{A'}+\cos\theta\Bra{B'}\right)\Ket{A'}\\
=&\cos2\theta\Ket{A'}-\sin2\theta\Ket{B'},
\end{split}
\end{equation}
\begin{equation}
\begin{split}
G'\Ket{B'}=&-Z_{0,H}'Z_{f}'\Ket{B'}\\
=&-\left(I^{\otimes(n+2^t+1)}-2\Ket{h}\Bra{h}-2\Ket{h^{\perp}}\Bra{h^{\perp}}\right)\Ket{B'}\\
=&-\Ket{B'}+2\left(\sin\theta\Ket{A'}+\cos\theta\Ket{B'}\right)\left(\sin\theta\Bra{A'}+\cos\theta\Bra{B'}\right)\Ket{B'}\\
=&\sin2\theta\Ket{A'}+\cos2\theta\Ket{B'}.
\end{split}
\end{equation}

Then we have
\begin{equation}\label{(G')^{l}}
(G')^{l}|h\rangle=\cos{((2l+1)\theta)}|B'\rangle+\sin{((2l+1)\theta)}|A'\rangle.
\end{equation}

In fact, Eq. (\ref{(G')^{l}}) can be proved by induction as follows.
When $l=1$, we obtain
\begin{equation}
\begin{split}
G'|h\rangle&=G'(\sin{\theta}|A'\rangle+\cos{\theta}|B'\rangle)\\
                  &=\sin{\theta}G'|A'\rangle+\cos{\theta}G'|B'\rangle\\
                  &=\sin{\theta}(\cos{2\theta}|A'\rangle-\sin{2\theta}|B'\rangle)+\cos{\theta}(\sin{2\theta}|A'\rangle+\cos{2\theta}|B'\rangle)\\
                  &=(\sin{\theta}\cos{2\theta}+\cos{\theta}\sin{2\theta})|A'\rangle+(\cos{\theta}\cos{2\theta}-\sin{\theta}\sin{2\theta})|B'\rangle\\                
                  &=\sin{3\theta}|A'\rangle+\cos{3\theta}|B'\rangle.
\end{split}
\end{equation}
Suppose that Eq. (\ref{(G')^{l}}) holds for $l=m$, i.e.
\begin{equation}
(G')^m|h\rangle=\sin{((2m+1)\theta)}|A\rangle+\cos{((2m+1)\theta)}|B\rangle.
\end{equation}
Then when $l=m+1$, we can get
\begin{align}
(G')^{m+1}|h\rangle=&(G')(G')^m|h\rangle\notag\\
                             =&(G')\left(\sin{((2m+1)\theta)}|A'\rangle+\cos{((2m+1)\theta)}|B'\rangle\right)\notag\\
                             =&\sin{((2m+1)\theta)}(G')|A'\rangle+\cos{((2m+1)\theta)}(G')|B'\rangle\notag\\
                             =&\sin{((2m+1)\theta)}(\cos{2\theta}|A'\rangle-\sin{2\theta}|B'\rangle)+\cos{((2m+1)\theta)}(\sin{2\theta}|A'\rangle+\cos{2\theta}|B'\rangle)\\
                             =&\left(\sin{((2m+1)\theta)}\cos{2\theta}+\cos{((2m+1)\theta)}\sin{2\theta}\right)|A'\rangle\notag\\&+\left(\cos{((2m+1)\theta)}\cos{2\theta}-\sin{((2m+1)\theta)}\sin{2\theta}\right)|B'\rangle\notag\\
                             =&\sin{((2(m+1)+1)\theta)}|A'\rangle+\cos{((2(m+1)+1)\theta)}|B'\rangle\notag.
\end{align}
Therefore, Eq. (\ref{(G')^{l}}) holds.

Our goal is to make $\sin{((2l+1)\theta)}\approx 1$,  that is $(2l+1)\theta\approx\frac{\pi}{2}$, so $l\approx\frac{\pi}{4\theta}-\frac{1}{2}$.

Due  to   $\theta=\arcsin {\sqrt{a/N}}$,  we obtain $l=\left\lfloor \pi/\left(4 \arcsin {\sqrt{a/N}}\right)\right\rfloor$, and the success probability of Algorithm \ref{DMA.alg} is:
\begin{equation}
\sin^2{\left(\left(2\left\lfloor \pi/\left(4 \arcsin {\sqrt{a/N}}\right)\right\rfloor+1\right)\arcsin {\sqrt{a/N}}\right)}.
\end{equation}

\end{proof}

Now, we  give a proposition for the equivalent representation of  $Z_{f}'$ defined in Eq. (\ref{Z_{f}'}).

\begin{proposition}\label{pp_Zf}
 $Z_{f}'$ defined in Eq. (\ref{Z_{f}'})  can  be represented as:
\begin{equation}
\begin{split}
Z_{f}'=F^{\dagger}\left(I^{\otimes \left(2^t+n\right)}\otimes Z\right)F,
\end{split}
\end{equation}
where $F$ is defined in Eq. (\ref{F}).
\end{proposition}
\begin{proof}

By means of Eq. (\ref{O_{f_w}A_w}), Eq. (\ref{U}) and Eq. (\ref{F}), we can obtain
\begin{equation}\label{F_uwab}
\begin{split}
F\Ket{u,w}\Ket{a}\Ket{b}
=&\left(I^{\otimes(n-t)}\otimes U\right)\left(\prod\nolimits_{w\in\{0,1\}^t}\left(O_{f_{w}}\otimes I^{\otimes \left(2^t-w\right)}\right)\right)\Ket{u,w}\Ket{a}\Ket{b}\\
=&\left(I^{\otimes(n-t)}\otimes U\right)\Ket{u,w}\ket{a_{0^t}\oplus f_{0^t}(u)}\Ket{a_{0^{t-1}1}\oplus f_{0^{t-1}1}(u)}\cdots\Ket{a_{1^t}\oplus f_{1^t}(u)}\Ket{b}\\
=&\Ket{u,w}\ket{a_{0^t}\oplus f_{0^t}(u)}\Ket{a_{0^{t-1}1}\oplus f_{0^{t-1}1}(u)}\cdots\Ket{a_{1^t}\oplus f_{1^t}(u)}\Ket{b\oplus f_w(u)},
\end{split}
\end{equation}
where   $u\in\{0,1\}^{n-t}$, $w\in\{0,1\}^{t}$, $a=a_{0^t}a_{0^{t-1}1}\cdots a_{1^t}\in \{0,1\}^{2^t}$ and $b\in \{0,1\}$.

With Eq. (\ref{F_uwab}), we have
\begin{equation}\label{FZF}
\begin{split}
&F^{\dagger}\left(I^{\otimes \left(2^t+n\right)}\otimes Z\right)F\Ket{u,w}\Ket{a}\Ket{b}\\
=&F^{\dagger}\left(I^{\otimes \left(2^t+n\right)}\otimes Z\right)\Ket{u,w}\ket{a_{0^t}\oplus f_{0^t}(u)}\Ket{a_{0^{t-1}1}\oplus f_{0^{t-1}1}(u)}\cdots\Ket{a_{1^t}\oplus f_{1^t}(u)}\Ket{b\oplus f_w(u)}\\
=&(-1)^{b\oplus f_w(u)}F^{\dagger}\Ket{u,w}\ket{a_{0^t}\oplus f_{0^t}(u)}\ket{a_{0^{t-1}1}\oplus f_{0^{t-1}1}(u)}\cdots\ket{a_{1^t}\oplus f_{1^t}(u)}\ket{b\oplus f_w(u)}\\
=&(-1)^{b\oplus f_w(u)}\Ket{u,w,a,b},
\end{split}
\end{equation}
where   $u\in\{0,1\}^{n-t}$, $w\in\{0,1\}^{t}$, $a=a_{0^t}a_{0^{t-1}1}\cdots a_{1^t}\in \{0,1\}^{2^t}$ and $b\in \{0,1\}$.

From Eq. (\ref{FZF}) it follows that  $Z_{f}'$ defined in Eq. (\ref{Z_{f}'})  can be represented as $F^{\dagger}\left(I^{\otimes \left(2^t+n\right)}\otimes Z\right)F$.
\end{proof}


Next, we  give the  proposition for the equivalent representation of  $Z_{0,H}'$ defined in Eq. (\ref{Z_{0,H}'}).

\begin{proposition}\label{pp_Z0H}
 $Z_{0,H}'$  defined in Eq. (\ref{Z_{0,H}'})  
can  be represented as:
\begin{equation}
\begin{split}
Z_{0,H}'=\left(H^{\otimes n}\otimes I^{\otimes \left(2^t+1\right)}\right)Z_0'\left(H^{\otimes n}\otimes I^{\otimes \left(2^t+1\right)}\right),
\end{split}
\end{equation}
where    $Z_0'$ is defined in Eq. (\ref{Z_0'}).
\end{proposition}
\begin{proof}
First we have
\begin{equation}\label{D_uwbcd}
\begin{split}
D\Ket{u,w}\Ket{b,c}\Ket{d}
=&\left(I^{\otimes (n-t)}\otimes V_2\otimes I^{\otimes\left(2^t-1\right)}\right)\left(V_1\otimes I^{\otimes 2^t}\right)\Ket{u,w}\Ket{b,c}\Ket{d}\\
=&\left(I^{\otimes (n-t)}\otimes V_2\otimes I^{\otimes\left(2^t-1\right)}\right)\Ket{u,w}\Ket{b\oplus {\rm OR}(u),c}\Ket{d}\\
=&\Ket{u,w}\Ket{b\oplus {\rm OR}(u),c\oplus\left(\lnot {\rm OR}(w(b\oplus {\rm OR}(u)))\right)}\Ket{d},
\end{split}
\end{equation}
where $D$ is defined in Eq. (\ref{D}), $u\in\{0,1\}^{n-t}$, $w\in\{0,1\}^{t}$, $b\in \{0,1\}$, $c\in \{0,1\}$  and  $d\in \{0,1\}^{2^t-1}$.

Therefore,
\begin{equation}
\begin{split}
Z_{0}'\Ket{u,w}\Ket{b,c}\Ket{d}
=&D^{\dagger}\left(I^{\otimes (n+1)}\otimes Z\otimes I^{\otimes\left(2^t-1\right)}\right)D\Ket{u,w}\Ket{b,c}\Ket{d}\\
=&D^{\dagger}\left(I^{\otimes (n+1)}\otimes Z\otimes I^{\otimes\left(2^t-1\right)}\right)\Ket{u,w}\Ket{b\oplus {\rm OR}(u),c\oplus\left(\lnot {\rm OR}(w(b\oplus {\rm OR}(u)))\right)}\Ket{d}\\
=&(-1)^{\left(c\oplus\left(\lnot {\rm OR}(w(b\oplus {\rm OR}(u)))\right)\right)}D^{\dagger}\Ket{u,w}\Ket{b\oplus {\rm OR}(u),c\oplus\left(\lnot {\rm OR}(w(b\oplus {\rm OR}(u)))\right)}\Ket{d}\\
=&(-1)^{ g(u,w,b,c)}\Ket{u,w,b,c,d},
\end{split}
\end{equation}
where the function $g$ is defined in Eq. (\ref{g(u,w,b,c)}),  $u\in\{0,1\}^{n-t}$, $w\in\{0,1\}^{t}$, $b\in \{0,1\}$, $c\in \{0,1\}$  and  $d\in \{0,1\}^{2^t-1}$.
Equivalently, it is
\begin{equation}\label{Z_0'2}
\begin{split}
Z_0'=&I^{\otimes(n+2^t+1)}-2\sum_{g(u,w,b,c)=1}\Ket{u,w,b,c,d}\Bra{u,w,b,c,d}\\
=&I^{\otimes(n+2^t+1)}-2\Ket{0^{n+2^t+1}}\Bra{0^{n+2^t+1}}-2\sum_{\substack{uwbcd\neq 0^{n+2^t+1}\\g(u,w,b,c)=1}}\Ket{u,w,b,c,d}\Bra{u,w,b,c,d},
\end{split}
\end{equation}
where   $u\in\{0,1\}^{n-t}$, $w\in\{0,1\}^{t}$, $b\in \{0,1\}$, $c\in \{0,1\}$  and  $d\in \{0,1\}^{2^t-1}$.

According to Eq. (\ref{g(u,w,b,c)}), when $uwbcd\neq 0^{n+2^t+1}$ and $g(u,w,b,c)=1$, we have $bcd\neq 0^{n+2^t+1}$ and $g(u,w,b,c)=1$. On the other hand, from $bcd\neq 0^{n+2^t+1}$,  we get $uwbcd\neq 0^{n+2^t+1}$. Furthermore, Eq. (\ref{Z_0'2}) can be equivalently expressed as follows:
\begin{equation}
Z_0'=I^{\otimes(n+2^t+1)}-2\Ket{0^{n+2^t+1}}\Bra{0^{n+2^t+1}}-2\sum_{\substack{bcd\neq 0^{2^t+1}\\g(u,w,b,c)=1}}\Ket{u,w,b,c,d}\Bra{u,w,b,c,d},
\end{equation}
where   $u\in\{0,1\}^{n-t}$, $w\in\{0,1\}^{t}$, $b\in \{0,1\}$, $c\in \{0,1\}$  and  $d\in \{0,1\}^{2^t-1}$.

As a result, we can obtain
\begin{align}
&\left(H^{\otimes n}\otimes I^{\otimes \left(2^t+1\right)}\right)Z_0'\left(H^{\otimes n}\otimes I^{\otimes \left(2^t+1\right)}\right)\\
=&
\left(H^{\otimes n}\otimes I^{\otimes \left(2^t+1\right)}\right)\left(I^{\otimes(n+2^t+1)}-2\Ket{0^{n+2^t+1}}\Bra{0^{n+2^t+1}}
-2\sum_{\substack{bcd\neq 0^{2^t+1}\\ g(u,w,b,c)=1 }}\Ket{u,w,b,c,d}\Bra{u,w,b,c,d}\right)\\&\left(H^{\otimes n}\otimes I^{\otimes \left(2^t+1\right)}\right)\\
=&
I^{\otimes(n+2^t+1)}-2\left(H^{\otimes n}\otimes I^{\otimes \left(2^t+1\right)}\right)\Ket{0^{n+2^t+1}}\Bra{0^{n+2^t+1}}\left(H^{\otimes n}\otimes I^{\otimes \left(2^t+1\right)}\right)
\\&-2\sum_{\substack{bcd\neq 0^{2^t+1}\\ g(u,w,b,c)=1 }}\left(H^{\otimes n}\otimes I^{\otimes \left(2^t+1\right)}\right)\Ket{u,w,b,c,d}\Bra{u,w,b,c,d}\left(H^{\otimes n}\otimes I^{\otimes \left(2^t+1\right)}\right)\\
=&I^{\otimes(n+2^t+1)}-2\Ket{h}\Bra{h}-2\sum_{\substack{bcd\neq 0^{2^t+1}\\ g(u,w,b,c)=1 }}H^{\otimes n}\Ket{u,w}\Ket{b,c,d}\Bra{u,w}H^{\otimes n}\Bra{b,c,d}\\
=&I^{\otimes(n+2^t+1)}-2\Ket{h}\Bra{h}-2\Ket{h^{\perp}}\Bra{h^{\perp}}\\
=&Z_{0,H}'.
\end{align}
So, $Z_{0,H}'$   
can  be represented as  $\left(H^{\otimes n}\otimes I^{\otimes \left(2^t+1\right)}\right)Z_0'\left(H^{\otimes n}\otimes I^{\otimes \left(2^t+1\right)}\right)$.
\end{proof}

\subsection{Space complexity analysis of the iterated operator in our algorithm}

In this subsection, we analyze the space complexity of  the largest single computing node in the iterated operator of
 Algorithm \ref{DMA.alg}. 

\begin{proposition}
 The number of qubits required for the largest single computing node in   Algorithm \ref{DMA.alg} is $\max\{n-t+1, 2^t+t+1\}$.
\end{proposition}
\begin{proof}

In Fig. \ref{FD operator -G'},  i.e.,  the further decomposition of the circuit diagram for the itertated operator $-G'$, the number of qubits required to realize a single computing node corresponding to each sub-oracle $O_{f_w}$ $(w\in\{0,1\}^t)$ is $n-t+1$, and the number of qubits required to realize a single computing node corresponding to the operator $U$ is $2^t+t+1$.  The numbers of qubits required to implement a single computing node corresponding to the operators $V_1$ and $V_2$ are $n-t+1$ and $t+2$, respectively. Thus, the number of qubits required to implement a largest single computing node corresponding to  $G'$ is $\max\{n-t+1, 2^t+t+1\}$.


We note that $H^{\otimes n}$ can be decomposed into the tensor product of $H^{\otimes (n-t)}$ and $H^{\otimes t}$ and can be implemented on two different computing nodes.  $G'$ can be decomposed into circuits  as Fig. \ref{FD operator -G'}. As a result, the number of qubits required to implement a largest single computing node in Algorithm \ref{DMA.alg} is $\max\{n-t+1, 2^t+t+1\}$.
\end{proof}

\subsection{Communication complexity analysis of algorithm}

The quantum communication complexity of Algorithm \ref{DMA.alg} is analyzed in the following. 

\begin{proposition}
Let Boolean function $f:\{0,1\}^n\rightarrow\{0,1\}$, $f$ is divided into $2^t$ subfunctions $f_w:\{0,1\}^{n-t}\rightarrow\{0,1\}$, where $f_w(u)=f(uw)$, $u \in \{0,1\}^{n-t}$, $w \in \{0,1\}^t$, $1\leq t<n$.  The quantum communication complexity of Algorithm \ref{DMA.alg} is $O\left(\sqrt{2^n}\left(2^t(n-t+1)+t\right)\right)$.
\end{proposition}
\begin{proof}
First,  in Fig.  \ref{FD operator -G'}, the $n-t$  qubits of  Node 1 are transmitted to   $O_{f_{0^t}}$ of Node 3 as its control qubits, then the $n-t$ control qubits of $O_{f_{0^t}}$  are transmitted to  $O_{f_{0^{t-1}1}}$ of Node 4 as its control qubits,  similarly  the $n-t$ control qubits of  $O_{f_{0^{t-1}1}}$ of Node 4 are transmitted to  $O_{f_{0^{t-2}10}}$ of Node 5 as its control qubits, and so on. Finally, the $n-t$ control qubits of  $O_{f_{1^{t-1}0}}$ of Node $2^t+1$ are transmitted to  $O_{f_{1^ t}}$ of Node $2^t+2$ as its control qubits. The quantum communication complexity of the process is $O\left(2^t(n-t)\right)$.

Next, after obtaining the state $\ket{\phi_1}$ in Algorithm \ref{DMA.alg},
the $t$ qubits of  Node 2 are transmitted to   $U$ of  Node $2^t+3$,  the  target qubit of   $O_{f_{0^{t}}}$ of Node 3 is transmitted to   $U$ of  Node $2^t+3$, the  target qubit of   $O_{f_{0^{t-1}1}}$ of Node 4 is transmitted to   $U$ of  Node $2^t+3$, and so on. Finally,  the  target qubit  of  $O_{f_{1^ t}}$ of Node $2^t+2$ is transmitted to  $U$ of Node $2^t+3$. The quantum communication complexity of the process is $O \left(2^t+t\right)$. 
Similarly, the above analysis shows that in Algorithm \ref{DMA.alg}, the quantum communication complexity of the process from state $\ket{\phi_2}$ to state $\ket{\phi_3}$ is $O\left(2^t(n-t+1)+t\right)$.

Then, in Algorithm \ref{DMA.alg}, from state $\ket{\phi_4}$ to state $\ket{\phi_5}$, $n-t$ qubits of Node 1 are transmitted to $V_1$ of Node 3. From state $\ket{\phi_5}$ to state $\ket{\phi_6}$, $t+1$ qubits of Node 2 and Node 3 are transmitted to $V_2$ of Node 4. The complexity of quantum communication from state $\ket{\phi_4}$ to state $\ket{\phi_6}$ in Algorithm \ref{DMA.alg}  is $O \left(n\right)$. Similarly, the quantum communication complexity of the process from state $\ket{\phi_6}$ to state $\ket{\phi_7}$ is $O \left(n\right)$.

Finally, based on the above analysis, the quantum communication complexity of one iteration run of algorithm \ref{DMA.alg} is $O\left(2^t(n-t+1)+n+t\right)$. Since Algorithm \ref{DMA.alg} needs to be run $O\left(\sqrt{2^n}\right)$ times iteratively, the quantum communication complexity of Algorithm \ref{DMA.alg} is $O\left(\sqrt{2^n}\left(2^t(n-t+1)+n+t\right)\right)$. 
So, the quantum communication complexity of Algorithm \ref{DMA.alg} is $O\left(n^2\sqrt{2^n}\right)$.
\end{proof}

\section{Distributed exact multi-objective quantum search algorithm}\label{Sec4}

In this section, we first describe  the design of  a distributed exact multi-objective quantum search algorithm, and then give its correctness, space  complexity and communication complexity analyses.

\subsection{Design of algorithm}

In the following, we  describe the operators used in  the distributed exact multi-objective quantum search algorithm, i.e., Algorithm \ref{DEMA.alg}.

Define the operator $S_{f}(\phi): \{0,1\}^{n+2^t+1}\rightarrow  \{0,1\}^{n+2^t+1}$ as:
\begin{equation}\label{S_{f}}
\begin{split}
S_{f}(\phi)\Ket{u,w,a,b}=e^{\mathrm{i}\phi\cdot (b\oplus f_w(u))}\Ket{u,w,a,b},
\end{split}
\end{equation}
where   $u\in\{0,1\}^{n-t}$, $w\in\{0,1\}^{t}$, $a=a_{0^t}a_{0^{t-1}1}\cdots a_{1^t}\in \{0,1\}^{2^t}$ and $b\in \{0,1\}$.

Define the operator $S_{0,H}(\phi): \{0,1\}^{n+2^t+1}\rightarrow  \{0,1\}^{n+2^t+1}$ as:
\begin{equation}\label{S_{0,H}}
\begin{split}
S_{0,H}(\phi)
=&I^{\otimes(n+2^t+1)}+(e^{\mathrm{i}\phi}-1)\Ket{h}\Bra{h}+(e^{\mathrm{i}\phi}-1)\Ket{h^{\perp}}\Bra{h^{\perp}},
\end{split}
\end{equation}
where $\Ket{h}$ is denoted in Eq. (\ref{h}), $\Ket{h^{\perp}}$ is denoted in Eq. (\ref{h^perp}).

Define the operator $Q: \{0,1\}^{n+2^t+1}\rightarrow  \{0,1\}^{n+2^t+1}$ as:
\begin{equation}\label{Q}
\begin{split}
Q=-S_{0,H}(\phi)S_{f}(\phi).
\end{split}
\end{equation}

\begin{figure}[H]
  \includegraphics[width=\textwidth]{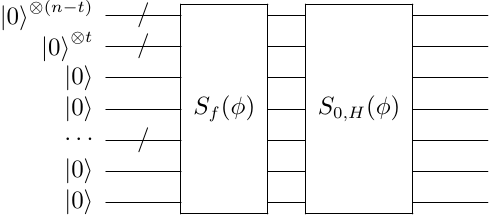}
  \caption{The circuit for  operator -$Q$.}
  \label{operator -Q)}
\end{figure}

In fact,  both $S_{f}(\phi)$ and  $S_{0,H}(\phi)$ can be decomposed into a number of smaller and  physically realizable operators. 
In the following, we introduce some operators related to the decomposition representation of  $S_{f}(\phi)$.

Define the operator $E(\phi): \{0,1\}\rightarrow \{0,1\}$   as:
\begin{equation}
\begin{split}
		\label{E_phi}
		E(\phi)\Ket{x}=
		\begin{cases}
		\Ket{x}\text{, } & x=0\text{;}\\
		e^{\mathrm{i}\phi}\Ket{x}\text{, } & x=1\text{.}
          \end{cases}
\end{split}
\end{equation}
In fact, $S_{f}(\phi)$ defined in Eq. (\ref{S_{f}}) can be decomposed  by the following form:
\begin{equation}\label{S_{f}DE}
\begin{split}
S_{f}(\phi)=&F^{\dagger}\left(I^{\otimes \left(2^t+n\right)}\otimes E(\phi)\right)F,
\end{split}
\end{equation}
where  $F$ is defined in Eq. (\ref{F}). This result will be proved in Proposition \ref{pp_Sf} of Subsection \ref{Correctness analysis of  algorithm exact}.

\begin{figure}[H]
  \centerline{\includegraphics[width=0.7\textwidth]{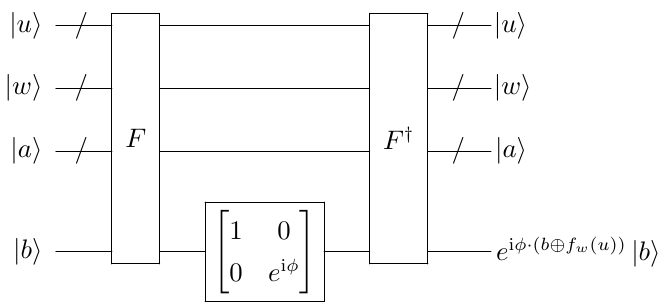}}
  \caption{The circuit for  operator $S_f(\phi)$.}
\end{figure}

In the following, we introduce some operators related to the decomposition of $S_{0,H}(\phi)$.

Define the operator $S_0(\phi): \{0,1\}^{n+2^t+1}\rightarrow  \{0,1\}^{n+2^t+1}$  as:
\begin{equation}\label{S_0}
\begin{split}
S_0(\phi)=D^{\dagger}\left(I^{\otimes (n+1)}\otimes E(\phi)\otimes I^{\otimes\left(2^t-1\right)}\right)D,
\end{split}
\end{equation}
where $D$ is defined in Eq. (\ref{D}).

\begin{figure}[H]
  \centerline{\includegraphics[width=0.7\textwidth]{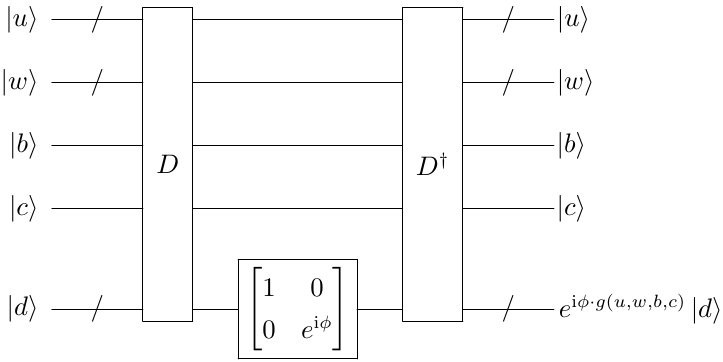}}
  \caption{The circuit for  operator $S_0(\phi)$.}
\end{figure}

 $S_{0,H}(\phi)$ defined in Eq. (\ref{S_{0,H}}) can be decomposed as follows:
\begin{equation}\label{S_{0,H}DE}
\begin{split}
S_{0,H}(\phi)=\left(H^{\otimes n}\otimes I^{\otimes \left(2^t+1\right)}\right)S_0(\phi)\left(H^{\otimes n}\otimes I^{\otimes \left(2^t+1\right)}\right).
\end{split}
\end{equation}
The proof will be deferred to Proposition \ref{pp_S0H} of Subsection \ref{Correctness analysis of  algorithm exact}.

\begin{figure}[H]
  \centerline{\includegraphics[width=0.7\textwidth]{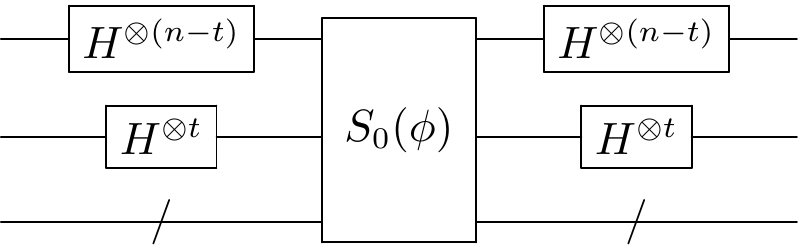}}
  \caption{The circuit for  operator $S_0H(\phi)$.}
\end{figure}

Now, we give the further decomposition circuit for  -$Q$.

\begin{figure}[H]
  \includegraphics[width=\textwidth]{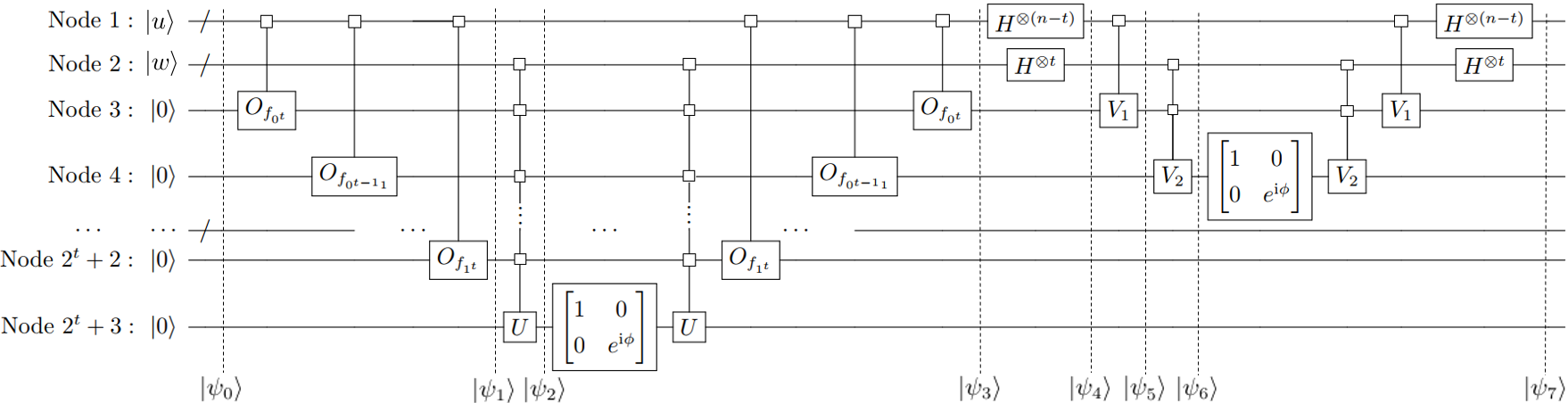}
  \caption{The further decomposition circuit for    -$Q$.}
  \label{FD operator -Q}
\end{figure}

\begin{algorithm}
\caption{Distributed exact multi-objective quantum search algorithm}
\label{DEMA.alg}
\begin{algorithmic}
\STATE \textbf{Input}: $f:\{0,1\}^n \rightarrow\{0,1\}$ with $|\{x\in\{0,1\}^n| f(x)=1\}|=a\geq 1$, $N=2^n$.
\STATE \textbf{Initialization}: $n>4$, $1<t<\log_2n-1$, $\theta= \arcsin {\sqrt{a/N}}$, $K=\lfloor(\pi/2-\theta)/(2\theta)\rfloor$, $\phi=2\arcsin\left(\sin\left(\frac{\pi}{4K+6}\right) / \sin \theta\right)$.
\STATE \textbf{Output}: The string $x$ such that $f(x)=1$ with certainty.
\STATE \textbf{Procedure:}

\STATE \textbf{1.} $\left(H^{\otimes (n-t)}\otimes H^{\otimes t} \otimes I^{\otimes 2^t}\otimes I\right)|0\rangle^{\otimes (n-t)}|0\rangle^{\otimes t}|0\rangle^{\otimes 2^t}|0\rangle=\frac{1}{\sqrt{2^n}}\sum\limits_{x\in\{0 , 1\}^n}|x\rangle|0\rangle^{\otimes (2^t+1)}$.

\STATE \textbf{2.} $Q$ is performed with $K+1$ times, where $Q=-S_{0,H}(\phi)S_{f}(\phi)$, $S_f(\phi)$ is defined in Eq. ($\ref{S_{f}}$), and $S_{0,H}(\phi)$ is defined in Eq. ($\ref{S_{0,H}}$).

\STATE \textbf{3.} Measure the first $n$ qubits of the resulting state.
\end{algorithmic}
\end{algorithm}






\begin{figure*}[h]
  \includegraphics[width=\textwidth]{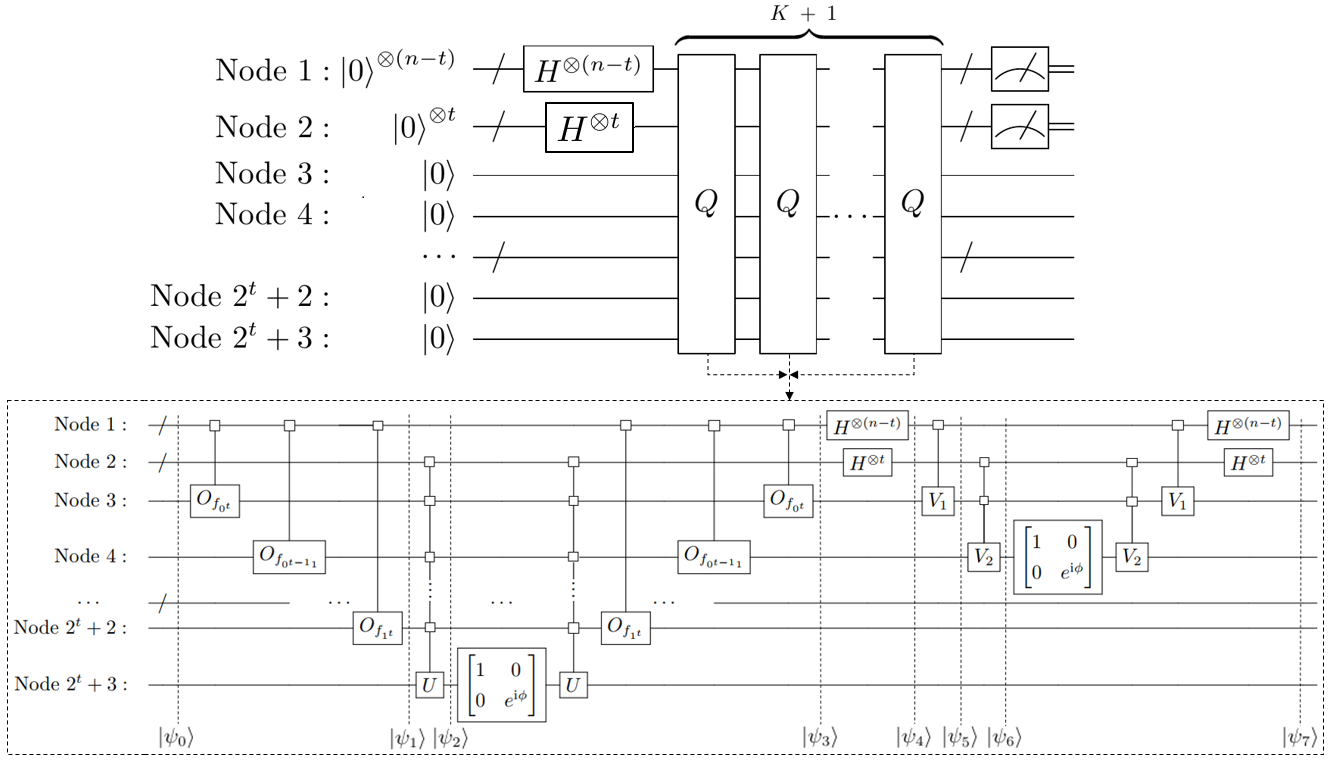}
  \caption{The circuit for  Algorithm \ref{DEMA.alg}.}
\end{figure*}


\subsection{Correctness analysis of  algorithm} \label{Correctness analysis of  algorithm exact}

In this subsection, we first give the relevant lemmas and a theorem proving the correctness of Algorithm \ref{DEMA.alg}, and then two propositions are given for the equivalent representations of  $S_{f}(\phi)$ and $S_{0,H}(\phi)$. First, we give five lemmas related to proving the correctness of Algorithm \ref{DEMA.alg}.

\begin{lemma}\label{lem_MQ}
Given a matrix $M_Q=-e^{\mathrm{i}\phi}\begin{bmatrix}
1+(e^{\mathrm{i}\phi}-1)\sin^2\theta & (1-e^{-\mathrm{i}\phi})\sin\theta\cos\theta \\
(e^{\mathrm{i}\phi}-1)\sin\theta\cos\theta &1-(1-e^{-\mathrm{i}\phi})\sin^2\theta
\end{bmatrix}$, $M_Q$ can be equivalently represented as
\begin{equation}
 M_Q=-e^{\mathrm{i}\phi}\left[\cos\left(\frac{\alpha}{2}\right)I+\mathrm{i}\sin\left(\frac{\alpha}{2}\right)\left(n_xX+n_yY+n_zZ\right)\right],
 \end{equation}
where $\alpha=4\beta$, $\sin\beta=\sin\left(\frac{\phi}{2}\right)\sin\theta$, $\theta=\arcsin(\sqrt{a/N})$, 
$n_x=\frac{\cos\theta}{\cos\beta}\cos\left(\frac{\phi}{2}\right), n_y=\frac{\cos\theta}{\cos\beta}\sin\left(\frac{\phi}{2}\right), n_z=\frac{\cos\theta}{\cos\beta}\cos\left(\frac{\phi}{2}\right)\tan\theta,
$
and
$
I=\begin{bmatrix}
1 & 0 \\
0 & 1\end{bmatrix},
X=\begin{bmatrix}
0 & 1\\
1 & 0\end{bmatrix},
Y=\begin{bmatrix}
0 & -\mathrm{i}\\
\mathrm{i} & 0\end{bmatrix},
Z=\begin{bmatrix}
1 & 0\\
0 & -1\end{bmatrix}
$.
\end{lemma}
\begin{proof}
Let
\begin{equation}
\begin{split}
\begin{bmatrix}
q_{11} &q_{12}  \\
q_{21}  &q_{22} 
\end{bmatrix}=&-e^{\mathrm{i}\phi}\left[\cos\left(\frac{\alpha}{2}\right)I+\mathrm{i}\sin\left(\frac{\alpha}{2}\right)\left(n_xX+n_yY+n_zZ\right)\right]\\
=&-e^{\mathrm{i}\phi}\begin{bmatrix}
\cos\left(\frac{\alpha}{2}\right)+\mathrm{i}\sin\left(\frac{\alpha}{2}\right)n_z &\mathrm{i}\sin\left(\frac{\alpha}{2}\right)n_x+\sin\left(\frac{\alpha}{2}\right)n_y \\
\mathrm{i}\sin\left(\frac{\alpha}{2}\right)n_x-\sin\left(\frac{\alpha}{2}\right)n_y  &\cos\left(\frac{\alpha}{2}\right)-\mathrm{i}\sin\left(\frac{\alpha}{2}\right)n_z 
\end{bmatrix},
\end{split}
\end{equation}
where $\alpha=4\beta$, $\sin\beta=\sin\left(\frac{\phi}{2}\right)\sin\theta$, $\theta=\arcsin(\sqrt{a/N})$ and
\begin{equation}
\begin{split}
n_x=\frac{\cos\theta}{\cos\beta}\cos\left(\frac{\phi}{2}\right), n_y=\frac{\cos\theta}{\cos\beta}\sin\left(\frac{\phi}{2}\right), n_z=\frac{\cos\theta}{\cos\beta}\cos\left(\frac{\phi}{2}\right)\tan\theta.
\end{split}
\end{equation}

In order to prove that $M_Q$ can be equivalently represented as 
\begin{equation}\label{M_QERP}
M_Q=-e^{\mathrm{i}\phi}\left[\cos\left(\frac{\alpha}{2}\right)I+\mathrm{i}\sin\left(\frac{\alpha}{2}\right)\left(n_xX+n_yY+n_zZ\right)\right], 
\end{equation}
it is equivalent to proving
\begin{equation}
\begin{bmatrix}
q_{11} &q_{12}  \\
q_{21}  &q_{22} 
\end{bmatrix}=\begin{bmatrix}
-e^{\mathrm{i}\phi}(1+(e^{\mathrm{i}\phi}-1)\sin^2\theta) & -e^{\mathrm{i}\phi}(1-e^{-\mathrm{i}\phi})\sin\theta\cos\theta \\
-e^{\mathrm{i}\phi}(e^{\mathrm{i}\phi}-1)\sin\theta\cos\theta & -e^{\mathrm{i}\phi}\left(1-(1-e^{-\mathrm{i}\phi})\sin^2\theta\right)
\end{bmatrix}.
\end{equation}

In fact, we have
\begin{equation}\label{Q11}
\begin{split}
q_{11}=&-e^{\mathrm{i}\phi}\left(\cos\left(\frac{\alpha}{2}\right)+\mathrm{i}\sin\left(\frac{\alpha}{2}\right)n_z\right)\\
=&-e^{\mathrm{i}\phi}\left(\cos\left(\frac{\alpha}{2}\right)+\mathrm{i}\sin\left(\frac{\alpha}{2}\right)\frac{\cos\theta}{\cos\beta}\cos\left(\frac{\phi}{2}\right)\tan\theta\right)\\
=&-e^{\mathrm{i}\phi}\left(\cos\left(2\beta\right)+\mathrm{i}\sin\left(2\beta\right)\frac{\sin\theta}{\cos\beta}\cos\left(\frac{\phi}{2}\right)\right)\\
=&-e^{\mathrm{i}\phi}\left(1-2\left(\sin\left(\frac{\phi}{2}\right)\sin\theta\right)^2+2\mathrm{i}\left(\sin\left(\frac{\phi}{2}\right)\sin\theta\right)\cos\left(\frac{\phi}{2}\right)\sin\theta\right)\\
=&-e^{\mathrm{i}\phi}\left(1-2\sin^2\left(\frac{\phi}{2}\right)\sin^2\theta+\mathrm{i}\sin\phi\sin^2\theta\right)\\
=&-e^{\mathrm{i}\phi}(1+(e^{\mathrm{i}\phi}-1)\sin^2\theta).
\end{split}
\end{equation}

Similarly, we have
\begin{align}
q_{12}=& -e^{\mathrm{i}\phi}(1-e^{-\mathrm{i}\phi})\sin\theta\cos\theta,\\
q_{21}=& -e^{\mathrm{i}\phi}(e^{\mathrm{i}\phi}-1)\sin\theta\cos\theta,\\
q_{22}=&-e^{\mathrm{i}\phi}\left(1-(1-e^{-\mathrm{i}\phi})\sin^2\theta\right).\label{Q22}
\end{align}

In conclusion, from Eq. (\ref{Q11}) to Eq.  (\ref{Q22}), it follows that Eq. (\ref{M_QERP}) holds.

\end{proof}

Suppose $\boldsymbol{r}=[x,y,z]^{\mathrm{T}}$ is a vector on $\mathbb{R}^3$. Then we  define
\begin{equation}
\begin{split}
\boldsymbol{r}*\boldsymbol{\sigma}=&x\sigma_1+y\sigma_2+z\sigma_3\\
=&\begin{bmatrix}
z & x-\mathrm{i}y  \\
x+\mathrm{i}y   & -z
\end{bmatrix},
\end{split}
\end{equation}
where $\sigma_1=\begin{bmatrix}
0 & 1  \\
1   & 0
\end{bmatrix}$, $\sigma_2=\begin{bmatrix}
0 & -\mathrm{i}  \\
\mathrm{i}   & 0
\end{bmatrix}$ and  $\sigma_3=\begin{bmatrix}
1 & 0  \\
0   & -1
\end{bmatrix}$ are Pauli matrices.

\begin{lemma}\label{lem_QRQHOMP}
Given  a vector $\boldsymbol{r}=[x,y,z]^{\mathrm{T}}$ on $\mathbb{R}^3$, a matrix 
$
H_{\boldsymbol{r}}
=\boldsymbol{r}*\boldsymbol{\sigma}
=\begin{bmatrix}
z & x-\mathrm{i}y  \\
x+\mathrm{i}y   & -z
\end{bmatrix}$,
 a matrix $M_Q=e^{\mathrm{i}\phi}
\begin{bmatrix}
u & v  \\
-\overline{v}  & \overline{u} 
\end{bmatrix}$, then the following matrix 
\begin{equation}
R_{M_Q}=
\begin{bmatrix}
\frac{1}{2}(u^2+\overline{u}^{2}-v^2-\overline{v}^{2})&\frac{\mathrm{i}}{2}(\overline{u}^{2}-u^2+\overline{v}^{2}-v^2)& -(uv+\overline{u}\overline{v})  \\
\frac{\mathrm{i}}{2}(u^2-\overline{u}^{2}+\overline{v}^{2}-v^{2}) & \frac{1}{2}(u^2+\overline{u}^{2}+v^{2}+\overline{v}^{2})&\mathrm{i}(\overline{u}\overline{v}-uv) \\
(\overline{u}v+u\overline{v})&\mathrm{i}(\overline{u}v-u\overline{v})&(u\overline{u}-v\overline{v})
\end{bmatrix}
\end{equation}
satisfies:
\begin{equation}
M_QH_{\boldsymbol{r}}{M_Q}^{-1}=(R_{M_Q}\boldsymbol{r})* \boldsymbol{\sigma},
\end{equation}
where $|u|^2+|v|^2=1$.
\end{lemma}\label{lem_QRQHOMP}
\begin{proof}

Let 
\begin{align}
M_QH_{\boldsymbol{r}}{M_Q}^{-1}&=
e^{\mathrm{i}\phi}
\begin{bmatrix}
u & v  \\
-\overline{v}  & \overline{u} 
\end{bmatrix}
\begin{bmatrix}
z & x-\mathrm{i}y  \\
x+\mathrm{i}y   & -z
\end{bmatrix}
e^{-\mathrm{i}\phi}
\begin{bmatrix}
\overline{u} & -v  \\
\overline{v}  & u 
\end{bmatrix}\label{QH1}
\\&=
\begin{bmatrix}
z' & x'-\mathrm{i}y'  \\
x'+\mathrm{i}y'   & -z'
\end{bmatrix}\label{QH2}
\\&=
\begin{bmatrix}
h'_{11} & h'_{12}  \\
h'_{21} & h'_{22} 
\end{bmatrix}\label{QH3}
\\&=H_{\boldsymbol{r'}}
\\&=\boldsymbol{r'}*\boldsymbol{\sigma}\label{QH4},
\end{align}
where $\boldsymbol{r'}=[x',y',z']^{\mathrm{T}}$ is a vector on $\mathbb{R}^3$.

By Eq. (\ref{QH1}), Eq. (\ref{QH2}) and Eq. (\ref{QH3}), we can obtain
\begin{equation}\label{RQ1}
\begin{cases}
x'=\dfrac{h'_{21}+h'_{12} }{2}=\frac{1}{2}(u^2+\overline{u}^{2}-v^2-\overline{v}^{2})x+\frac{\mathrm{i}}{2}(\overline{u}^{2}-u^2+\overline{v}^{2}-v^2)y-(uv+\overline{u}\overline{v})z,\\
y'=\dfrac{h'_{21}-h'_{12} }{2\mathrm{i}}=\frac{\mathrm{i}}{2}(u^2-\overline{u}^{2}+\overline{v}^{2}-v^{2})x+\frac{1}{2}(u^2+\overline{u}^{2}+v^{2}+\overline{v}^{2})y+\mathrm{i}(\overline{u}\overline{v}-uv)z,\\
z'=h'_{11}=-h'_{22}=(\overline{u}v+u\overline{v})x+\mathrm{i}(\overline{u}v-u\overline{v})y+(u\overline{u}-v\overline{v})z.
\end{cases}
\end{equation}

Thus, if the matrix $M_Q$ satisfies $M_QH_{\boldsymbol{r}}{M_Q}^{-1}=H_{\boldsymbol{r'}}$, then $R_{M_Q}$ also satisfies
\begin{equation}\label{RQr}
R_{M_Q}\boldsymbol{r}=\boldsymbol{r'}.
\end{equation}

With Eq. (\ref{QH1}), Eq. (\ref{QH4}) and Eq. (\ref{RQr}), we can get
\begin{equation}
M_QH_{\boldsymbol{r}}{M_Q}^{-1}=(R_{M_Q}\boldsymbol{r})* \boldsymbol{\sigma}.
\end{equation}

According to Eq. (\ref{RQ1}),  $R_{M_Q}$ is given as:
\begin{equation}
\begin{split}
R_{M_Q}=&
\begin{bmatrix}
r_{11}&r_{12}&r_{13}  \\
r_{21} & r_{22} &r_{23}\\
r_{31}&r_{32}&r_{33}
\end{bmatrix}\\
=&
\begin{bmatrix}
\frac{1}{2}(u^2+\overline{u}^{2}-v^2-\overline{v}^{2})&\frac{\mathrm{i}}{2}(\overline{u}^{2}-u^2+\overline{v}^{2}-v^2)& -(uv+\overline{u}\overline{v})  \\
\frac{\mathrm{i}}{2}(u^2-\overline{u}^{2}+\overline{v}^{2}-v^{2}) & \frac{1}{2}(u^2+\overline{u}^{2}+v^{2}+\overline{v}^{2})&\mathrm{i}(\overline{u}\overline{v}-uv) \\
(\overline{u}v+u\overline{v})&\mathrm{i}(\overline{u}v-u\overline{v})&(u\overline{u}-v\overline{v})
\end{bmatrix}.\label{RQ2}
\end{split}
\end{equation}
In particular, if $u = 1+(e^{\mathrm{i}\phi}-1)\sin^2\theta$ and $v = (1-e^{-\mathrm{i}\phi})\sin\theta\cos\theta$, i.e., 
\begin{equation}
M_Q =-e^{\mathrm{i}\phi} \begin{bmatrix}
1+(e^{\mathrm{i}\phi}-1)\sin^2\theta &(1-e^{-\mathrm{i}\phi})\sin\theta\cos\theta \\
(e^{\mathrm{i}\phi}-1)\sin\theta\cos\theta &1-(1-e^{-\mathrm{i}\phi})\sin^2\theta
\end{bmatrix},
\end{equation}
then we can derive the following equations:
\begin{align}
r_{11}=&\cos\phi\left[\cos^2(2\theta)\cos\phi+\sin^2(2\theta)\right]+\cos(2\theta)\sin^2\phi,\notag\\
r_{12}=&-\cos(2\theta)\cos\phi\sin\phi+\left[\cos^2\left(\frac{\phi}{2}\right)-\cos(4\theta)\sin^2\left(\frac{\phi}{2}\right)\right]\sin\phi,\notag\\
r_{13}=&-\sin(4\theta)\sin^2\left(\frac{\phi}{2}\right),\notag\\
r_{21}=&\cos\phi\sin\phi\left[\cos(2\theta)-1\right],\notag\\
r_{22}=&\cos^2\phi+\cos(2\theta)\sin^2\phi,\\
r_{23}=&\sin(2\theta)\sin\phi,\notag\\
r_{31}=&-\cos\phi\sin(4\theta)\sin^2\left(\frac{\phi}{2}\right)+\sin(2\theta)\sin^2\phi,\notag\\
r_{32}=&-\cos\phi\sin(2\theta)\sin\phi-\sin(4\theta)\sin^2\left(\frac{\phi}{2}\right)\sin\phi,\notag\\
r_{33}=&\cos^2(2\theta)+\cos\phi\sin^2(2\theta).\notag
\end{align}

\end{proof}

Suppose $\boldsymbol{\hat{n}}=\left[n_x,n_y,n_z\right]^{\mathrm{T}}=\left[\sin\theta'\cos\varphi',\sin\theta'\sin\varphi',\cos\theta'\right]^{\mathrm{T}}
$ is a vector on $\mathbb{R}^3$.
Then we define
\begin{equation}\label{RnDE}
R_{\hat{n}}(-\alpha)=R_{z}(\varphi')R_{y}(\theta')R_{z}(-\alpha)R_{y}(-\theta')R_{z}(-\varphi'),
\end{equation}
where
\begin{align}
R_{y}(\theta')=&\cos\left(\frac{\theta'}{2}\right)I-\mathrm{i}\sin\left(\frac{\theta'}{2}\right)Y=\begin{bmatrix}\cos\frac{\theta'}{2} &-\sin\frac{\theta'}{2} \label{R_{y}}\\
-\sin\frac{\theta'}{2} &\cos\frac{\theta'}{2} \end{bmatrix},\\
R_{z}(\varphi')=&\cos\left(\frac{\varphi'}{2}\right)I-\mathrm{i}\sin\left(\frac{\varphi'}{2}\right)Z=\begin{bmatrix}e^{-\mathrm{i}\varphi'/2} &0\\
0 &e^{\mathrm{i}\varphi'/2} \end{bmatrix}.\label{R_{z}}
\end{align}
In light of Eq. (\ref{RQ2}), we know that the matrix $R_{y}(\theta')$ in Eq. (\ref{R_{y}}) corresponds to the following matrix:
\begin{equation}
\begin{bmatrix}
\cos\theta'&0&\sin\theta'  \\
0 & 1 &0\\
-\sin\theta&0&\cos\theta'
\end{bmatrix},
\end{equation}
and the matrix $R_{z}(\varphi')$ in Eq. (\ref{R_{z}}) corresponds to the following matrix:
\begin{equation}
\begin{bmatrix}
\cos\varphi'&-\sin\varphi'&0  \\
\sin\varphi' & \cos\varphi' &0\\
0&0&1
\end{bmatrix}.
\end{equation}


\begin{lemma}\label{lem_Q3DR}
Given a matrix
$M_Q=-e^{\mathrm{i}\phi}\left[\cos\left(\frac{\alpha}{2}\right)I+\mathrm{i}\sin\left(\frac{\alpha}{2}\right)\left(n_xX+n_yY+n_zZ\right)\right]$, 
$M_Q$ can be equivalently represented as
\begin{equation}
M_Q
=-e^{\mathrm{i}\phi}R_{\hat{n}}(-\alpha),
\end{equation}
where $\boldsymbol{\hat{n}}=\left[n_x,n_y,n_z\right]^{\mathrm{T}}$,  $R_{\hat{n}}(-\alpha)$ is defined in Eq. (\ref{RnDE}).
\end{lemma}
\begin{proof}


By using the following equations
\begin{align}
R_{y}(\theta')ZR_{y}(-\theta')=&\cos\theta'Z+\sin\theta'X,\\
R_{z}(\varphi')XR_{z}(-\varphi')=&\cos\varphi'X+\sin\varphi'Y,\\
R_{z}(\varphi')ZR_{z}(-\varphi')=&Z,
\end{align}
where $R_{y}(\theta')$ is defined in Eq. (\ref{R_{y}}), $R_{z}(\varphi')$ is defined in Eq. (\ref{R_{z}}),
we can obtain
\begin{equation}\label{Rn}
\begin{split}
R_{\hat{n}}(-\alpha)=&R_{z}(\varphi')R_{y}(\theta')R_{z}(-\alpha)R_{y}(-\theta')R_{z}(-\varphi')\\
=&R_{z}(\varphi')R_{y}(\theta')\left[\cos\left(\frac{\alpha}{2}\right)I+\mathrm{i}\sin\left(\frac{\alpha}{2}\right)Z\right]R_{y}(-\theta')R_{z}(-\varphi')\\=&\cos\left(\frac{\alpha}{2}\right)I+\mathrm{i}\sin\left(\frac{\alpha}{2}\right)R_{z}(\varphi')R_{y}(\theta')ZR_{y}(-\theta')R_{z}(-\varphi')\\
=&\cos\left(\frac{\alpha}{2}\right)I+\mathrm{i}\sin\left(\frac{\alpha}{2}\right)\left(\cos\theta'Z+\sin\theta'\cos\varphi'X+\sin\theta'\sin\varphi'Y\right)\\
=&\cos\left(\frac{\alpha}{2}\right)I+\mathrm{i}\sin\left(\frac{\alpha}{2}\right)\left(n_xX+n_yY+n_zZ\right).
\end{split}
\end{equation}

With Eq. (\ref{Rn}), we have
\begin{align}
M_Q=&-e^{\mathrm{i}\phi}\left[\cos\left(\frac{\alpha}{2}\right)I+\mathrm{i}\sin\left(\frac{\alpha}{2}\right)\left(n_xX+n_yY+n_zZ\right)\right]\\
=&-e^{\mathrm{i}\phi}R_{\hat{n}}(-\alpha).
\end{align}
\end{proof}

\begin{lemma}\label{lem_rpsi}
Given  a quantum state $\Ket{\psi}=(a+b\mathrm{i})\Ket{A'}+(c+d\mathrm{i})\Ket{B'}$,  a matrix $
M_Q=e^{\mathrm{i}\phi}
\begin{bmatrix}
u & v  \\
-\overline{v}  & \overline{u} 
\end{bmatrix}$,
then there exists a corresponding vector 
$
\boldsymbol{r_{\Ket{\psi}}}=\left[\Bra{\psi}X\Ket{\psi},\Bra{\psi}Y\Ket{\psi},\Bra{\psi}Z\Ket{\psi}\right]^{\mathrm{T}}
$,
together with a  corresponding matrix 
\begin{equation}\label{RQ2LEM4}
R_{M_Q}=\begin{bmatrix}
r_{11}&r_{12}&r_{13}  \\
r_{21} & r_{22} &r_{23}\\
r_{31}&r_{32}&r_{33}
\end{bmatrix}=
\begin{bmatrix}
\frac{1}{2}(u^2+\overline{u}^{2}-v^2-\overline{v}^{2})&\frac{\mathrm{i}}{2}(\overline{u}^{2}-u^2+\overline{v}^{2}-v^2)& -(uv+\overline{u}\overline{v})  \\
\frac{\mathrm{i}}{2}(u^2-\overline{u}^{2}+\overline{v}^{2}-v^{2}) & \frac{1}{2}(u^2+\overline{u}^{2}+v^{2}+\overline{v}^{2})&\mathrm{i}(\overline{u}\overline{v}-uv) \\
(\overline{u}v+u\overline{v})&\mathrm{i}(\overline{u}v-u\overline{v})&(u\overline{u}-v\overline{v})
\end{bmatrix} 
\end{equation}
 satisfying
\begin{equation}
\boldsymbol{r_{M_Q\Ket{\psi}}}=R_{M_Q}\boldsymbol{r_{\Ket{\psi}}},
\end{equation}
where $|u|^2+|v|^2=1$.
\end{lemma}
\begin{proof}

Since
\begin{equation}\label{rQ_psi}
\boldsymbol{r_{M_Q\Ket{\psi}}}=\left[\Bra{\psi}{M_Q}^{\dagger}X{M_Q}\Ket{\psi},\Bra{\psi}{M_Q}^{\dagger}Y{M_Q}\Ket{\psi},\Bra{\psi}{M_Q}^{\dagger}ZM_Q\Ket{\psi}\right]^{\mathrm{T}},
\end{equation}
and
\begin{equation}\label{Quv}
M_Q=e^{\mathrm{i}\phi}
\begin{bmatrix}
u & v  \\
-\overline{v}  & \overline{u} 
\end{bmatrix},
\end{equation}
we have
\begin{equation}\label{RQr_psi}
\begin{split}
R_{M_Q}\boldsymbol{r_{\Ket{\psi}}}=&
\begin{bmatrix}
r_{11}&r_{12}&r_{13}  \\
r_{21} & r_{22} &r_{23}\\
r_{31}&r_{32}&r_{33}
\end{bmatrix}\left[\Bra{\psi}X\Ket{\psi},\Bra{\psi}Y\Ket{\psi},\Bra{\psi}Z\Ket{\psi}\right]^{\mathrm{T}}\\
=&
\begin{bmatrix}
\Bra{\psi}\left(r_{11}X+r_{12}Y+r_{13}Z\right)\Ket{\psi}  \\
\Bra{\psi}\left(r_{21}X + r_{22}Y +r_{23}Z\right)\Ket{\psi}\\
\Bra{\psi}\left(r_{31}X+r_{32}Y+r_{33}Z\right)\Ket{\psi}
\end{bmatrix}.
\end{split}
\end{equation}

According to Eq. (\ref{RQ2LEM4}) and Eq. (\ref{Quv}), it can be calculated that
\begin{equation}\label{QuvRQ2}
\begin{cases}
{M_Q}^{\dagger}X{M_Q}=r_{11}X+r_{12}Y+r_{13}Z,\\
{M_Q}^{\dagger}Y{M_Q}=r_{21}X+r_{22}Y+r_{23}Z,\\
{M_Q}^{\dagger}Z{M_Q}=r_{31}X+r_{32}Y+r_{33}Z.
\end{cases}
\end{equation}
By using Eq.(\ref{rQ_psi}), Eq. (\ref{RQr_psi}) and  Eq. (\ref{QuvRQ2}), it can be obtained that
\begin{equation}
\boldsymbol{r_{M_Q\Ket{\psi}}}=R_{M_Q}\boldsymbol{r_{\Ket{\psi}}}.
\end{equation}
\end{proof}

Suppose $\boldsymbol{r_{\ket{x}}}=\left[x_1,x_2,x_3\right]^{\mathrm{T}}$ and $\boldsymbol{r_{\ket{y}}}=\left[y_1,y_2,y_3\right]^{\mathrm{T}}$ are two vectors on $\mathbb{R}^3$.  Then we can define:
\begin{equation}
\boldsymbol{r_{\ket{x}}}\cdot \boldsymbol{r_{\ket{y}}}=x_1y_1+x_2y_2+x_3y_3,
\end{equation}
\begin{equation}
\left|\left|\boldsymbol{r_{\ket{x}}}\right|\right|=x_1^2+x_2^2+x_3^2,
\end{equation}
\begin{equation}
\left|\left|\boldsymbol{r_{\ket{y}}}\right|\right|=y_1^2+y_2^2+y_3^2.
\end{equation}

\begin{lemma}\label{cos_omega}
Given vectors $\boldsymbol{\hat{n}}=\frac{\cos\theta}{\cos\beta}\left[\cos\left(\frac{\phi}{2}\right),\sin\left(\frac{\phi}{2}\right),\cos\left(\frac{\phi}{2}\right)\tan\theta\right]^{\mathrm{T}}$, $\boldsymbol{r_{\ket{h}}}=\left[\sin(2\theta),0,-\cos(2\theta)\right]^{\mathrm{T}}$, $\boldsymbol{r_{\ket{A'}}}=\left[0,0,1\right]^{\mathrm{T}}$, and 
\begin{equation}\label{r_ketp}
\boldsymbol{r_{\ket{p}}}=\frac{1}{1+\cos^2\left(\frac{\phi}{2}\right)\tan^2\theta}\left[\cos^2\left(\frac{\phi}{2}\right)\tan\theta,\sin\left(\frac{\phi}{2}\right)\cos\left(\frac{\phi}{2}\right)\tan\theta,\cos^2\left(\frac{\phi}{2}\right)\tan^2\theta\right]^{\mathrm{T}},
\end{equation}
denote
\begin{equation}
\cos\omega=\frac{\left(\boldsymbol{r_{\ket{h}}}-\boldsymbol{r_{\ket{p}}}\right)\cdot\left(\boldsymbol{r_{\ket{A'}}}-\boldsymbol{r_{\ket{p}}}\right)}{\left|\left|\boldsymbol{r_{\ket{h}}}-\boldsymbol{r_{\ket{p}}}\right|\right|\left|\left|\boldsymbol{r_{\ket{A'}}}-\boldsymbol{r_{\ket{p}}}\right|\right|},
\end{equation}
then 
\begin{align}\cos\omega
=\cos\left(2\arccos\left(\sin\left(\frac{\phi}{2}\right)\sin\theta\right)\right),
\end{align}
where
$
\theta=\arcsin(\sqrt{a/N}), \sin\beta=\sin\left(\frac{\phi}{2}\right)\sin\theta$.
\end{lemma}
\begin{proof}
First, as we know, the equation of a line passing through the origin point and parallel to the vector $\boldsymbol{\hat{n}}=\frac{\cos\theta}{\cos\beta}\left[\cos\left(\frac{\phi}{2}\right),\sin\left(\frac{\phi}{2}\right),\cos\left(\frac{\phi}{2}\right)\tan\theta\right]^{\mathrm{T}}$ is
\begin{equation}\label{line}
\frac{x}{\cos\left(\frac{\phi}{2}\right)}=\frac{y}{\sin\left(\frac{\phi}{2}\right)}=\frac{z}{\cos\left(\frac{\phi}{2}\right)\tan\theta}.
\end{equation}
Also, the equation of the plane passing through point $\boldsymbol{r_{\ket{A'}}}=[0,0,1]^{\mathrm{T}}$ and
orthogonal  to the same vector  $\boldsymbol{\hat{n}}$ is
\begin{equation}\label{plane}
x\cos\left(\frac{\phi}{2}\right)+y\sin\left(\frac{\phi}{2}\right)+(z-1)\cos\left(\frac{\phi}{2}\right)\tan\theta=0.
\end{equation}
By
\begin{equation}
\begin{cases}
\frac{x}{\cos\left(\dfrac{\phi}{2}\right)}=\frac{y}{\sin\left(\dfrac{\phi}{2}\right)}=\frac{z}{\cos\left(\dfrac{\phi}{2}\right)\tan\theta},\\
x\cos\left(\frac{\phi}{2}\right)+y\sin\left(\frac{\phi}{2}\right)+(z-1)\cos\left(\frac{\phi}{2}\right)\tan\theta=0,
\end{cases}
\end{equation}
we can get
\begin{equation}
\begin{cases}
x=\frac{1}{1+\cos^2\left(\dfrac{\phi}{2}\right)\tan^2\theta}\cos^2\left(\frac{\phi}{2}\right)\tan\theta,\\
y=\frac{1}{1+\cos^2\left(\dfrac{\phi}{2}\right)\tan^2\theta}\sin\left(\frac{\phi}{2}\right)\cos\left(\frac{\phi}{2}\right)\tan\theta,\\
z=\frac{1}{1+\cos^2\left(\dfrac{\phi}{2}\right)\tan^2\theta}\cos^2\left(\frac{\phi}{2}\right)\tan^2\theta.
\end{cases}
\end{equation}
Therefore, the intersection point of the line represented by Eq. (\ref{line}) and the plane by Eq. (\ref{plane}) is:
\begin{equation}\label{inpoint}
\frac{1}{1+\cos^2\left(\frac{\phi}{2}\right)\tan^2\theta}\left[\cos^2\left(\frac{\phi}{2}\right)\tan\theta,\sin\left(\frac{\phi}{2}\right)\cos\left(\frac{\phi}{2}\right)\tan\theta,\cos^2\left(\frac{\phi}{2}\right)\tan^2\theta\right]^{\mathrm{T}}.
\end{equation}
So, $\boldsymbol{r_{\ket{p}}}$ is the projection of $\boldsymbol{r_{\ket{A'}}}$ on $\boldsymbol{\hat{n}}$.

According to
\begin{equation}
\begin{split}
\boldsymbol{\hat{n}}\cdot \boldsymbol{r_{\ket{h}}}&=\frac{\cos\theta}{\cos\beta}\left[\cos\left(\frac{\phi}{2}\right),\sin\left(\frac{\phi}{2}\right),\cos\left(\frac{\phi}{2}\right)\tan\theta\right]^{\mathrm{T}}\cdot \left[\sin(2\theta),0,-\cos(2\theta)\right]^{\mathrm{T}}\\
&=\frac{\cos\theta}{\cos\beta}\cos\left(\frac{\phi}{2}\right)\tan\theta,
\end{split}
\end{equation}
\begin{equation}
\begin{split}
\boldsymbol{\hat{n}}\cdot \boldsymbol{r_{\ket{A'}}}&=\frac{\cos\theta}{\cos\beta}\left[\cos\left(\frac{\phi}{2}\right),\sin\left(\frac{\phi}{2}\right),\cos\left(\frac{\phi}{2}\right)\tan\theta\right]^{\mathrm{T}}\cdot \left[0,0,1\right]^{\mathrm{T}}\\
&=\frac{\cos\theta}{\cos\beta}\cos\left(\frac{\phi}{2}\right)\tan\theta,
\end{split}
\end{equation}
we have
$\boldsymbol{\hat{n}}\cdot \boldsymbol{r_{\ket{h}}}=\boldsymbol{\hat{n}}\cdot \boldsymbol{r_{\ket{A'}}}$, the projections of $\boldsymbol{r_{\ket{h}}}$ and $\boldsymbol{r_{\ket{A'}}}$ on $\boldsymbol{\hat{n}}$ are the same.

Therefore, there is:
\begin{equation}\label{rhdotrp}
\boldsymbol{r_{\ket{h}}}\cdot \boldsymbol{r_{\ket{p}}}= \boldsymbol{r_{\ket{A'}}}\cdot \boldsymbol{r_{\ket{p}}}=\boldsymbol{r_{\ket{p}}}\cdot \boldsymbol{r_{\ket{p}}}.
\end{equation}

Since
\begin{equation}
\begin{split}
\boldsymbol{r_{\ket{A'}}}\cdot \boldsymbol{r_{\ket{p}}}=&\left[0,0,1\right]^{\mathrm{T}}\cdot
\frac{1}{1+\cos^2\left(\frac{\phi}{2}\right)\tan^2\theta}\left[\cos^2\left(\frac{\phi}{2}\right)\tan\theta,\sin\left(\frac{\phi}{2}\right)\cos\left(\frac{\phi}{2}\right)\tan\theta,\cos^2\left(\frac{\phi}{2}\right)\tan^2\theta\right]^{\mathrm{T}}\\
=&\frac{\cos^2\left(\frac{\phi}{2}\right)\tan^2\theta}{1+\cos^2\left(\frac{\phi}{2}\right)\tan^2\theta},
\end{split}
\end{equation}
we have
\begin{equation}
\boldsymbol{r_{\ket{h}}}\cdot \boldsymbol{r_{\ket{p}}}= \boldsymbol{r_{\ket{A'}}}\cdot \boldsymbol{r_{\ket{p}}}=\boldsymbol{r_{\ket{p}}}\cdot \boldsymbol{r_{\ket{p}}}=\frac{\cos^2\left(\frac{\phi}{2}\right)\tan^2\theta}{1+\cos^2\left(\frac{\phi}{2}\right)\tan^2\theta}.
\end{equation}

Since $\boldsymbol{r_{\ket{A'}}}\cdot \boldsymbol{r_{\ket{p}}}=\boldsymbol{r_{\ket{p}}}\cdot \boldsymbol{r_{\ket{p}}}$, $\boldsymbol{r_{\ket{A'}}}\cdot \boldsymbol{r_{\ket{p}}}=\boldsymbol{r_{\ket{p}}}\cdot \boldsymbol{r_{\ket{A'}}}$, we have
$\boldsymbol{r_{\ket{p}}}\cdot \boldsymbol{r_{\ket{A'}}}=\boldsymbol{r_{\ket{p}}}\cdot \boldsymbol{r_{\ket{p}}}$.

Since $\left|\left|\boldsymbol{r_{\ket{h}}}\right|\right|=\left|\left|\boldsymbol{r_{\ket{A'}}}\right|\right|=1$, it follows:
\begin{equation}\label{rh-rp}
\begin{split}
\left|\left|\boldsymbol{r_{\ket{h}}}-\boldsymbol{r_{\ket{p}}}\right|\right|\left|\left|\boldsymbol{r_{\ket{A'}}}-\boldsymbol{r_{\ket{p}}}\right|\right|
=&\sqrt{1-\left|\left|\boldsymbol{r_{\ket{p}}}\right|\right|^2}\sqrt{1-\left|\left|\boldsymbol{r_{\ket{p}}}\right|\right|^2}\\
=&\frac{1}{1+\cos^2\left(\frac{\phi}{2}\right)\tan^2\theta}.
\end{split}
\end{equation}

Therefore, we have:
\begin{align}
\cos\omega=&\frac{\left(\boldsymbol{r_{\ket{h}}}-\boldsymbol{r_{\ket{p}}}\right)\cdot\left(\boldsymbol{r_{\ket{A'}}}-\boldsymbol{r_{\ket{p}}}\right)}{\left|\left|\boldsymbol{r_{\ket{h}}}-\boldsymbol{r_{\ket{p}}}\right|\right|\left|\left|\boldsymbol{r_{\ket{A'}}}-\boldsymbol{r_{\ket{p}}}\right|\right|}\\
=&\frac{\boldsymbol{r_{\ket{h}}}\cdot\boldsymbol{r_{\ket{A'}}}-\boldsymbol{r_{\ket{h}}}\cdot\boldsymbol{r_{\ket{p}}}-\boldsymbol{r_{\ket{p}}}\cdot\boldsymbol{r_{\ket{A'}}}+\boldsymbol{r_{\ket{p}}}\cdot\boldsymbol{r_{\ket{p}}}}{\left|\left|\boldsymbol{r_{\ket{h}}}-\boldsymbol{r_{\ket{p}}}\right|\right|\left|\left|\boldsymbol{r_{\ket{A'}}}-\boldsymbol{r_{\ket{p}}}\right|\right|}\\
=&\left(1+\cos^2\left(\frac{\phi}{2}\right)\tan^2\theta\right)\left(-\cos(2\theta)-\frac{\cos^2\left(\frac{\phi}{2}\right)\tan^2\theta}{1+\cos^2\left(\frac{\phi}{2}\right)\tan^2\theta}\right)\\
=&-\cos(2\theta)-(\cos(2\theta)+1)\cos^2\left(\frac{\phi}{2}\right)\tan^2\theta\\
=&-\cos(2\theta)-2\cos^2\theta\cos^2\left(\frac{\phi}{2}\right)\tan^2\theta\\
=&-\cos(2\theta)-2\cos^2\left(\frac{\phi}{2}\right)\sin^2\theta\\
=&-\cos(2\theta)-(1+\cos\phi)\sin^2\theta\\
=&-\cos(2\theta)-\sin^2\theta-\cos\phi\sin^2\theta\\
=&-(1-2\sin^2\theta)-\sin^2\theta-\cos\phi\sin^2\theta\\
=&\sin^2\theta-1-\cos\phi\sin^2\theta.
\end{align}
Since
\begin{align}
\cos\left(2\arccos\left(\sin\left(\frac{\phi}{2}\right)\sin\theta\right)\right)=&2\cos^2\left(\arccos\left(\sin\left(\frac{\phi}{2}\right)\sin\theta\right)\right)-1\\
=&2\left(\sin\left(\frac{\phi}{2}\right)\sin\theta\right)^2-1\\
=&2\sin^2\left(\frac{\phi}{2}\right)\sin^2\theta-1\\
=&(1-\cos\phi)\sin^2\theta-1\\
=&\sin^2\theta-1-\cos\phi\sin^2\theta,
\end{align}
we have
\begin{equation}\cos\omega
=\cos\left(2\arccos\left(\sin\left(\frac{\phi}{2}\right)\sin\theta\right)\right).
\end{equation}
\end{proof}

Now, we give the theorem proving the correctness of Algorithm \ref{DEMA.alg}.

\begin{theorem}\label{Thm_DEMA2} 
Let Boolean function $f:\{0,1\}^n \rightarrow\{0,1\}$  with $|\{x\in\{0,1\}^n| f(x)=1\}|=a\geq 1$, $N=2^n$. Then the string $x$ with $f(x)=1$ can be exactly attained by  Algorithm \ref{DEMA.alg} querying $K+1$ times, where $K=\lfloor(\pi/2-\theta)/(2\theta)\rfloor$, $\theta= \arcsin {\sqrt{a/N}}$.
\end{theorem}

\begin{proof}

By combining Eq. (\ref{S_{f}}),  Eq. (\ref{S_{0,H}}) and  Eq. (\ref{Q}),  we have
\begin{equation}
\begin{split}
Q\Ket{A'}=&-S_{0,H}(\phi)S_f(\phi)\Ket{A'}\\
=&-\left(I^{\otimes(n+2^t+1)}+(e^{\mathrm{i}\phi}-1)\Ket{h}\Bra{h}+(e^{\mathrm{i}\phi}-1)\Ket{h^{\perp}}\Bra{h^{\perp}}\right)e^{\mathrm{i}\phi}\Ket{A'}\\
=&-e^{\mathrm{i}\phi}\Ket{A'}-e^{\mathrm{i}\phi}(e^{\mathrm{i}\phi}-1)\left(\sin\theta\Ket{A'}+\cos\theta\Ket{B'}\right)\left(\sin\theta\Bra{A'}+\cos\theta\Bra{B'}\right)\Ket{A'}\\
=&-e^{\mathrm{i}\phi}(1+(e^{\mathrm{i}\phi}-1)\sin^2\theta)\Ket{A'}-e^{\mathrm{i}\phi}(e^{\mathrm{i}\phi}-1)\sin\theta\cos\theta\Ket{B'},
\end{split}
\end{equation}
\begin{equation}
\begin{split}
Q\Ket{B'}=&-S_{0,H}(\phi)S_f(\phi)\Ket{B'}\\
=&-\left(I^{\otimes(n+2^t+1)}+(e^{\mathrm{i}\phi}-1)\Ket{h}\Bra{h}+(e^{\mathrm{i}\phi}-1)\Ket{h^{\perp}}\Bra{h^{\perp}}\right)\Ket{B'}\\
=&-\Ket{B'}-(e^{\mathrm{i}\phi}-1)\left(\sin\theta\Ket{A'}+\cos\theta\Ket{B'}\right)\left(\sin\theta\Bra{A'}+\cos\theta\Bra{B'}\right)\Ket{B'}\\
=&-e^{\mathrm{i}\phi}(1-e^{-\mathrm{i}\phi})\sin\theta\cos\theta\Ket{A'}-e^{\mathrm{i}\phi}\left(1-(1-e^{-\mathrm{i}\phi})\sin^2\theta\right)\Ket{B'}.
\end{split}
\end{equation}


 Denote
\begin{align}
M_Q=-e^{\mathrm{i}\phi}&\begin{bmatrix}
1+(e^{\mathrm{i}\phi}-1)\sin^2\theta & (1-e^{-\mathrm{i}\phi})\sin\theta\cos\theta \\
(e^{\mathrm{i}\phi}-1)\sin\theta\cos\theta & 1-(1-e^{-\mathrm{i}\phi})\sin^2\theta
\end{bmatrix}.\label{Q1}
\end{align}
According to Lemma \ref{lem_MQ}, we have 
\begin{equation}\label{M_QERP2}
M_Q=-e^{\mathrm{i}\phi}\left[\cos\left(\frac{\alpha}{2}\right)I+\mathrm{i}\sin\left(\frac{\alpha}{2}\right)\left(n_xX+n_yY+n_zZ\right)\right],
\end{equation}
where $\alpha=4\beta$, $\sin\beta=\sin\left(\frac{\phi}{2}\right)\sin\theta$, $\theta=\arcsin(\sqrt{a/N})$, 
\begin{equation}
\begin{split}
n_x=\frac{\cos\theta}{\cos\beta}\cos\left(\frac{\phi}{2}\right), n_y=\frac{\cos\theta}{\cos\beta}\sin\left(\frac{\phi}{2}\right), n_z=\frac{\cos\theta}{\cos\beta}\cos\left(\frac{\phi}{2}\right)\tan\theta,
\end{split}
\end{equation}
and
\begin{equation}
\begin{split}
I=\begin{bmatrix}
1 & 0 \\
0 & 1\end{bmatrix},
X=\begin{bmatrix}
0 & 1\\
1 & 0\end{bmatrix},
Y=\begin{bmatrix}
0 & -\mathrm{i}\\
\mathrm{i} & 0\end{bmatrix},
Z=\begin{bmatrix}
1 & 0\\
0 & -1\end{bmatrix}.
\end{split}
\end{equation}
In addition, by Lemma \ref{lem_Q3DR},  we have 
\begin{equation}\label{MQRn}
M_Q
=-e^{\mathrm{i}\phi}R_{\hat{n}}(-\alpha),
\end{equation}
 where $\boldsymbol{\hat{n}}=\left[n_x,n_y,n_z\right]^{\mathrm{T}}$,  $R_{\hat{n}}(-\alpha)$ is defined in Eq. (\ref{RnDE}). 



 Given the quantum state $\Ket{\psi}=(a+b\mathrm{i})\Ket{A'}+(c+d\mathrm{i})\Ket{B'}$, by using  Lemma \ref{lem_rpsi}, we have  
\begin{equation}\label{r_psi}
\begin{split}
\boldsymbol{r_{\Ket{\psi}}}=&\left[\Bra{\psi}X\Ket{\psi},\Bra{\psi}Y\Ket{\psi},\Bra{\psi}Z\Ket{\psi}\right]^{\mathrm{T}}\\
=&\left[2(ac+bd),2(ad-bc),a^2+b^2-c^2-d^2\right]^{\mathrm{T}}.
\end{split}
\end{equation}
So, for the initial state $\Ket{h}=\sin\theta\Ket{A'}+\cos\theta\Ket{B'}$,  it holds that
\begin{align}
\boldsymbol{r_{\ket{h}}}=&\left[\sin(2\theta),0,-\cos(2\theta)\right]^{\mathrm{T}},
\end{align}
and for the target state $\Ket{A'}$, it corresponds to 
\begin{align}
\boldsymbol{r_{\ket{A'}}}=&\left[0,0,1\right]^{\mathrm{T}}.
\end{align}

Denote $\boldsymbol{r_{\ket{p}}}$ as the projection of $\boldsymbol{r_{\ket{A'}}}$  on $\boldsymbol{\hat{n}}=\frac{\cos\theta}{\cos\beta}\left[\cos\left(\frac{\phi}{2}\right),\sin\left(\frac{\phi}{2}\right),\cos\left(\frac{\phi}{2}\right)\tan\theta\right]^{\mathrm{T}}$, and 
\begin{equation}
\cos\omega=\frac{(\boldsymbol{r_{\ket{h}}}-\boldsymbol{r_{\ket{p}}})\cdot(\boldsymbol{r_{\ket{A'}}}-\boldsymbol{r_{\ket{p}}})}{\left|\boldsymbol{r_{\ket{h}}}-\boldsymbol{r_{\ket{p}}}\right|\left|\boldsymbol{r_{\ket{A'}}}-\boldsymbol{r_{\ket{p}}}\right|},
\end{equation}
where $\omega$ is the angle we have to rotate within a given number of iterations.

By Lemma \ref{cos_omega}, we have
\begin{align}
\cos\omega
=&\cos\left(2\arccos\left(\sin\left(\frac{\phi}{2}\right)\sin\theta\right)\right)\label{cosomega2},
\end{align}
and then 
\begin{equation}\label{omega}
\begin{split}
\omega=&2\arccos\left(\sin\left(\frac{\phi}{2}\right)\sin\theta\right)\\
=&2\left[\frac{\pi}{2}-\arcsin\left(\sin\left(\frac{\phi}{2}\right)\sin\theta\right)\right].
\end{split}
\end{equation}

Finding the target state is exactly achieved if angle $-\omega$ is $K+1$ times of 
$-\alpha$, that is,
\begin{equation}\label{omegaK+1}
\begin{split}
\omega=&\left(K+1\right)\alpha\\
=&4\left(K+1\right)\arcsin\left(\sin\left(\frac{\phi}{2}\right)\sin\theta\right).
\end{split}
\end{equation}

By combining Eq. (\ref{omega}) and Eq. (\ref{omegaK+1}), it holds that
\begin{equation}\label{phi}
\begin{split}
\phi=2\arcsin\left(\sin\left(\frac{\pi}{4K+6}\right)/\sin\theta\right).
\end{split}
\end{equation}

 In Algorithm  \ref{DEMA.alg}, we set $K=\lfloor(\pi/2-\theta)/(2\theta)\rfloor$. 
Since we have initialized  the angle $\phi$ in Algorithm  \ref{DEMA.alg} as $2\arcsin\left(\sin\left(\frac{\pi}{4K+6}\right)/\sin\theta\right)$, the string $x\in A$ can be obtained without error after  measuring the first $n$ qubits of  $Q^{K+1}\Ket{h}$.
\end{proof}

\begin{figure}[H]
  \centerline{\includegraphics[width=0.5\textwidth]{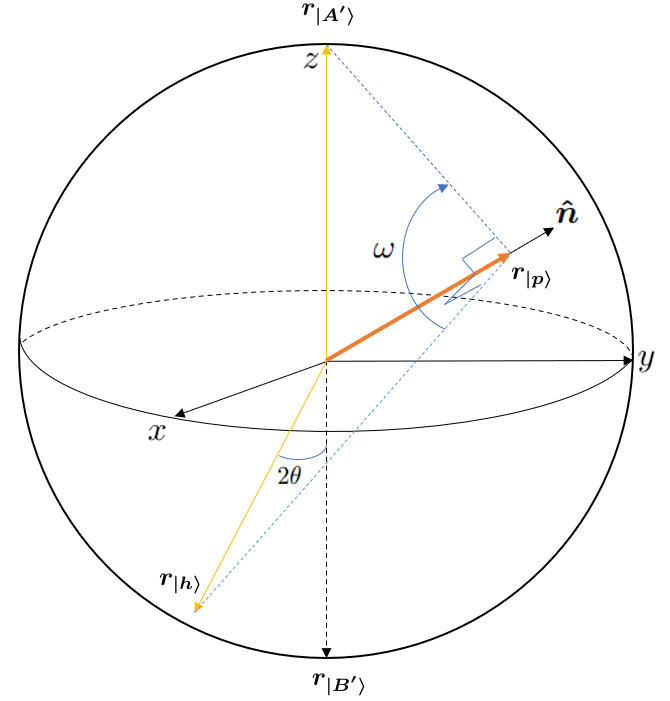}}
  \caption{Visualization of the distributed exact multi-objective quantum search algorithm.}
\end{figure}


\begin{proposition}\label{pp_Sf}
The operator $S_{f}(\phi)$ defined in Eq. (\ref{S_{f}})  can  be represented as:
\begin{equation}
\begin{split}
S_{f}(\phi)=F^{\dagger}\left(I^{\otimes \left(2^t+n\right)}\otimes E(\phi)\right)F,
\end{split}
\end{equation}
where $F$ is defined in Eq. (\ref{F}),  and $E(\phi)$ is defined in Eq. (\ref{E_phi}).
\end{proposition}

\begin{proof}
With Eq. (\ref{F_uwab}), we have
\begin{equation}\label{FEF}
\begin{split}
&F^{\dagger}\left(I^{\otimes \left(2^t+n\right)}\otimes E(\phi)\right)F\Ket{u,w}\Ket{a}\Ket{b}\\
=&F^{\dagger}\left(I^{\otimes \left(2^t+n\right)}\otimes E(\phi)\right)\Ket{u,w}\ket{a_{0^t}\oplus f_{0^t}(u)}\Ket{a_{0^{t-1}1}\oplus f_{0^{t-1}1}(u)}\cdots\Ket{a_{1^t}\oplus f_{1^t}(u)}\Ket{b\oplus f_w(u)}\\
=&e^{\mathrm{i}\phi\cdot (b\oplus f_w(u))}F^{\dagger}\Ket{u,w}\ket{a_{0^t}\oplus f_{0^t}(u)}\ket{a_{0^{t-1}1}\oplus f_{0^{t-1}1}(u)}\cdots\ket{a_{1^t}\oplus f_{1^t}(u)}\ket{b\oplus f_w(u)}\\
=&e^{\mathrm{i}\phi\cdot (b\oplus f_w(u))}\Ket{u,w,a,b}\\
=&S_{f}(\phi)\Ket{u,w,a,b},
\end{split}
\end{equation}
where   $u\in\{0,1\}^{n-t}$, $w\in\{0,1\}^{t}$, $a=a_{0^t}a_{0^{t-1}1}\cdots a_{1^t}\in \{0,1\}^{2^t}$ and $b\in \{0,1\}$.


\end{proof}



\begin{proposition}\label{pp_S0H}
r $S_{0,H}(\phi)$ defined in Eq. (\ref{S_{0,H}})  
can  be represented as:
\begin{equation}
\begin{split}
S_{0,H}(\phi)
=&\left(H^{\otimes n}\otimes I^{\otimes \left(2^t+1\right)}\right)S_0(\phi)\left(H^{\otimes n}\otimes I^{\otimes \left(2^t+1\right)}\right),
\end{split}
\end{equation}
where $S_0(\phi)$ is defined in Eq. (\ref{S_0}).
\end{proposition}
\begin{proof}
From Eq. (\ref{D_uwbcd}) and Eq. (\ref{S_0}), it follows that
\begin{equation}
\begin{split}
&S_{0}(\phi)\Ket{u,w}\Ket{b,c}\Ket{d}\\
=&D^{\dagger}\left(I^{\otimes (n+1)}\otimes E(\phi)\otimes I^{\otimes\left(2^t-1\right)}\right)D\Ket{u,w}\Ket{b,c}\Ket{d}\\
=&D^{\dagger}\left(I^{\otimes (n+1)}\otimes E(\phi)\otimes I^{\otimes\left(2^t-1\right)}\right)\Ket{u,w}\Ket{b\oplus {\rm OR}(u),c\oplus\left(\lnot {\rm OR}(w(b\oplus {\rm OR}(u)))\right)}\Ket{d}\\
=&e^{\mathrm{i}\phi\cdot \left(c\oplus\left(\lnot {\rm OR}(w(b\oplus {\rm OR}(u)))\right)\right)}D^{\dagger}\Ket{u,w}\Ket{b\oplus {\rm OR}(u),c\oplus\left(\lnot {\rm OR}(w(b\oplus {\rm OR}(u)))\right)}\Ket{d}\\
=&e^{\mathrm{i}\phi\cdot g(u,w,b,c)}\Ket{u,w,b,c,d},
\end{split}
\end{equation}
and therefore
\begin{equation}\label{S_02}
\begin{split}
S_0(\phi)=&I^{\otimes(n+2^t+1)}+\sum_{g(u,w,b,c)=1}(e^{\mathrm{i}\phi}-1)\Ket{u,w,b,c,d}\Bra{u,w,b,c,d}\\
=&I^{\otimes(n+2^t+1)}+(e^{\mathrm{i}\phi}-1)\Ket{0^{n+2^t+1}}\Bra{0^{n+2^t+1}}+\sum_{\substack{uwbcd\neq 0^{n+2^t+1}\\g(u,w,b,c)=1}}(e^{\mathrm{i}\phi}-1)\Ket{u,w,b,c,d}\Bra{u,w,b,c,d},
\end{split}
\end{equation}
where   $u\in\{0,1\}^{n-t}$, $w\in\{0,1\}^{t}$, $b\in \{0,1\}$, $c\in \{0,1\}$  and  $d\in \{0,1\}^{2^t-1}$.

According to Eq. (\ref{g(u,w,b,c)}), when $uwbcd\neq 0^{n+2^t+1}$ and $g(u,w,b,c)=1$, we have $bcd\neq 0^{n+2^t+1}$ and $g(u,w,b,c)=1$. On the other hand, from $bcd\neq 0^{n+2^t+1}$,  we get $uwbcd\neq 0^{n+2^t+1}$. Thus, Eq. (\ref{S_02}) can be equivalently expressed as follows:
\begin{equation}
S_0(\phi)=I^{\otimes(n+2^t+1)}+(e^{\mathrm{i}\phi}-1)\Ket{0^{n+2^t+1}}\Bra{0^{n+2^t+1}}+\sum_{\substack{bcd\neq 0^{2^t+1}\\g(u,w,b,c)=1}}(e^{\mathrm{i}\phi}-1)\Ket{u,w,b,c,d}\Bra{u,w,b,c,d},
\end{equation}
where   $u\in\{0,1\}^{n-t}$, $w\in\{0,1\}^{t}$, $b\in \{0,1\}$, $c\in \{0,1\}$  and  $d\in \{0,1\}^{2^t-1}$.

By the above equations we can further derive 
\begin{align}
&\left(H^{\otimes n}\otimes I^{\otimes \left(2^t+1\right)}\right)S_0(\phi)\left(H^{\otimes n}\otimes I^{\otimes \left(2^t+1\right)}\right)\notag\\
=&
\left(H^{\otimes n}\otimes I^{\otimes \left(2^t+1\right)}\right)\Bigg(I^{\otimes(n+2^t+1)}+(e^{\mathrm{i}\phi}-1)\Ket{0^{n+2^t+1}}\Bra{0^{n+2^t+1}}
\notag\\&\left.+\sum_{\substack{bcd\neq 0^{2^t+1}\\ g(u,w,b,c)=1 }}(e^{\mathrm{i}\phi}-1)\Ket{u,w,b,c,d}\Bra{u,w,b,c,d}\right)\left(H^{\otimes n}\otimes I^{\otimes \left(2^t+1\right)}\right)\notag\\
=&
I^{\otimes(n+2^t+1)}+(e^{\mathrm{i}\phi}-1)\left(H^{\otimes n}\otimes I^{\otimes \left(2^t+1\right)}\right)\Ket{0^{n+2^t+1}}\Bra{0^{n+2^t+1}}\left(H^{\otimes n}\otimes I^{\otimes \left(2^t+1\right)}\right)\notag
\\&+\sum_{\substack{bcd\neq 0^{2^t+1}\\ g(u,w,b,c)=1 }}(e^{\mathrm{i}\phi}-1)\left(H^{\otimes n}\otimes I^{\otimes \left(2^t+1\right)}\right)\Ket{u,w,b,c,d}\Bra{u,w,b,c,d}\left(H^{\otimes n}\otimes I^{\otimes \left(2^t+1\right)}\right)\\
=&I^{\otimes(n+2^t+1)}+(e^{\mathrm{i}\phi}-1)\Ket{h}\Bra{h}+(e^{\mathrm{i}\phi}-1)\Ket{h^{\perp}}\Bra{h^{\perp}}\notag\\
=&S_{0,H}(\phi)\notag.
\end{align}


\end{proof}

\subsection{Space and communication complexity analyses of algorithm\ref{DEMA.alg}}

The following two results and corresponding proofs are the same as those for Algorithm \ref{DMA.alg}. 

\begin{proposition}
The number of qubits required for the largest single computing node in  
 Algorithm \ref{DEMA.alg} is $\max\{n-t+1, 2^t+t+1\}$.
\end{proposition}






\begin{proposition}
Let Boolean function $f:\{0,1\}^n\rightarrow\{0,1\}$, $f$ is divided into $2^t$ subfunctions $f_w:\{0,1\}^{n-t}\rightarrow\{0,1\}$, where $f_w(u)=f(uw)$, $u \in \{0,1\}^{n-t}$, $w \in \{0,1\}^t$, $1\leq t<n$.  The quantum communication complexity of Algorithm \ref{DEMA.alg} is $O\left(\sqrt{2^n}\left(2^t(n-t+1)+t\right)\right)$.
\end{proposition}

\section{Comparisons of  algorithms}\label{Sec5}

 For an unstructured  search problem described by a Boolean function $f:\{0,1\}^n\rightarrow\{0,1\}$, 
 we compare  our algorithms with  Grover's algorithm \cite{grover_fast_1996} and the modified Grover’s  algorithm by Long \cite{Long2001} as well as the distributed Grover's algorithm by  Qiu et al \cite{Qiu24}.
 
In Grover's algorithm, the number of qubits required for  the largest single computing node is $n$, and in  our algorithms it is $\max\{n-t+1,2^t+t+1\}$.  
In addition, Grover's algorithm is not exact, but one of our  algorithms is exact.  
 is $O\left(n^2\sqrt{2^n}\right)$.   


In the modified Grover’s algorithm by Long, the number of qubits required for  the largest single computing node is also $n$, but that  in our algorithms is $\max\{n-t+1,2^t+t+1\}$.  

\begin{table}[h] 
	\centering
	\caption{Comparisons of our algorithms with other  algorithms in \cite{grover_fast_1996, Long2001,Qiu24}.}
	\begin{tabular}{*{4}{c}}
		\toprule
		                 Algorithms      &  \makecell[c]{Number of qubits} & \makecell[c]{Probability of\\ success}
		                 &  \makecell[c]{Quantum communication \\ complexity}\\
		\hline
		   Grover's algorithm  \cite{grover_fast_1996}       &        $n$ & \makecell[c]{High but \\ smaller than 1} 
		  & 0\\
		  \hline
		  \makecell[c]{Modified Grover's algorithm \\ by Long \cite{Long2001}}         &      $n$ & 1 
		  &    0  \\	
		  \hline	
		  
 \makecell[c]{Distributed Grover's algorithm \\  by Qiu et al  \cite{Qiu24} }      &        $n-t$ & \makecell[c]{High but \\ smaller than 1}  
		  & 0\\
		  \hline
  Algorithm  \ref{DMA.alg}     &   $\max\{n-t+1,2^t+t+1\}$ & \makecell[c]{High but \\ smaller than 1}
		  &    $O\left(n^2\sqrt{2^n}\right)$          \\
		  \hline
 Algorithm    \ref{DEMA.alg}     &   $\max\{n-t+1,2^t+t+1\}$ & 1  
&       $O\left(n^2\sqrt{2^n}\right)$        \\
		\toprule
	\end{tabular}
\end{table}


In the distributed Grover's algorithm  proposed by Qiu et al \cite{Qiu24},  the number of qubits required for  the largest single computing node is  $n-t$, but it may have some tiny error, while one of our algorithms \ref{DEMA.alg}  is exact.

Of course, 
both our algorithms require quantum communication and their quantum communication complexity is $O\left(n^2\sqrt{2^n}\right)$, but  Grover's algorithm and the modified Grover’s algorithm by Long as well as  the  distributed Grover's algorithm  proposed by Qiu et al.  does not require quantum communication. 

\section{Conclusion}\label{Sec6}

In this paper, motivated and inspired by the methods of distributed Grover's algorithm  and distributed Simon's algorithm of Qiu and Tan et al. \cite{Qiu24,Tan2022DQCSimon}, and by using the basic ideas and technical processes of Grover's algorithm \cite{grover_fast_1996} and the modified Grover's algorithm of Long \cite{grover_fast_1996, Long2001}, we have proposed a distributed multi-objective quantum search algorithm and a distributed exact multi-objective quantum search algorithm. 

Given an unstructured search problem described by Boolean function $f:\{0,1\}^n \rightarrow\{0,1\}$ with $|\{x\in\{0,1\}^n| f(x)=1\}|=a\geq 1$, for designing a distributed multi-objective quantum search algorithm,  we actually designed a distributed iterated operator $G'$.  More concretely, we decomposed  $f$ into $2^t$ sub-functions $f_w:\{0,1\}^{n-t}\rightarrow\{0,1\}$, where $w\in\{0,1\}^t$. Then, we united  $2^t$ query operators $O_{f_w}$ by employing  a unitary operator $U$, and combining with Pauli operator $Z$ to achieve the effect of the query operator $Z_f$ in Grover's algorithm. After that, 
we combined a unitary operator $V_1$ with another unitary operator $V_2$, and the Pauli operator $Z$ to achieve the effect of   diffusion operator  $H^{\otimes n}Z_0H^{\otimes n}$ in Grover's algorithm. 
By means of the modified Grover's algorithm of Long \cite{grover_fast_1996, Long2001}, we have further designed a distributed exact multi-objective quantum search algorithm and its method and process  is similar to the above one.

Compared to the previous related algorithms, our algorithms  have showed certain advantages. More specifically,  
Grover's algorithm \cite{grover_fast_1996}   and the modified Grover’s algorithm by Long  \cite{Long2001}   require that the number of qubits  for  computing node is $n$, and  both both our algorithms require  $\max\{n-t+1,2^t+t+1\}$;  the distributed Grover's algorithm  proposed by Qiu et al \cite{Qiu24} requires  the number of qubits  for  the largest single computing node is  $n-t$, but as Grover's algorithm it is not completely exact, while one of  our algorithms   is exact.
Of course,  none of Grover's algorithm \cite{grover_fast_1996}   and the modified Grover’s algorithm by Long  \cite{Long2001}  as well as the distributed Grover's algorithm  by Qiu et al \cite{Qiu24} need  quantum communication, but both our algorithms require quantum communication and their quantum communication complexity  is $O\left(n^2\sqrt{2^n}\right)$.   However, our algorithms still may  have  certain advantage of  physical realizability in the  Noisy Intermediate-Scale Quantum (NISQ) era.  Also, in a way, we have further exerted and developed  the technical methods of designing distributed quantum algorithms by Qiu et al. \cite{Qiu24,Tan2022DQCSimon}.


Finally, some  problems worthy of further consideration are to design  distributed quantum amplitude amplification algorithm and  distributed quantum random walks algorithm, by appropriately using these technical methods of distributed quantum algorithms in \cite{Qiu24,Tan2022DQCSimon} and this paper.

\section*{Acknowledgments}
This work is supported in part by the National Natural Science Foundation of China (Grant No. 61876195), the Natural Science Foundation of Guangdong Province of China (Grant No. 2017B030311011), the Shenzhen Science and Technology Program (Grant No. JCYJ20220818102003006), and the Guangdong Provincial Quantum Science Strategic Initiative (No. GDZX2200001).

\bibliographystyle{unsrt}

\end{document}